\documentclass[runningheads,envcountsame]{llncs} 

\usepackage{hyperref}

\usepackage{latexsym}
\usepackage{amsmath,amssymb,amstext}
\usepackage{wasysym}

\usepackage{cite}

\usepackage{graphicx}
\usepackage{stmaryrd}

\usepackage{tikz}
\usetikzlibrary{arrows,arrows.meta,automata,backgrounds,calc,chains,decorations.fractals,decorations.pathreplacing,fadings,fit,folding,mindmap,patterns,plotmarks,positioning,shadows,shapes.geometric,shapes.symbols,through,trees}

\colorlet{dblue}{blue!40!black}

\usepackage{textcomp}
\usepackage{enumitem}
\usepackage{verbatim}
\usepackage{float}
\usepackage{pdfpages}

\usepackage{wrapfig}

\usepackage{fontawesome}

\newcommand{\todo}[1]{{\em\color{blue}{\footnotesize{Todo:~#1}}}}

\spnewtheorem{notation}[theorem]{Notation}{\itshape}{\normalfont}

\newcommand{\funap}[2]{#1(#2)}

\renewcommand{\emptyset}{\varnothing}

\newcommand{\set}[1]{\{\,#1\,\}}

\tikzset{arsdefault/.style={thick,shorten >= 0mm, shorten <= .5mm,inner sep=.3mm,outer sep=0mm}}
\tikzset{anno/.style={scale=.85,blue}}
\tikzset{wave/.style={decorate,decoration={snake,amplitude=.4mm,segment length=2mm,post length=1mm}}}
\tikzset{bluered/.style={cblue,nodes={black}}}
\tikzset{redred/.style={cred,nodes={black},wave}}
\tikzset{smallCircle/.style={circle,fill=black,inner sep=0mm,outer sep=1mm,minimum size=1mm}}

\tikzset{default/.style={
  thick,
  every node/.style={circle},
  level distance=12mm, 
  inner sep=.5mm}}
\tikzset{smallCircle/.style={circle,fill=black,inner sep=0mm,outer sep=1mm,minimum size=1mm}}
\tikzset{paint/.style={very thick,draw=#1!50!black,fill=#1,opacity=.4}}
\tikzset{paintopaque/.style={very thick,draw=#1!50!black!60,fill=#1!60}}
\tikzset{loopl/.style={out=140,in=210,looseness=10}}
\tikzset{loopr/.style={out=-40,in=30,looseness=10}}
\tikzset{loopb/.style={out=-40-90,in=30-90,looseness=10}}
\tikzset{loopt/.style={out=-40+90,in=30+90,looseness=10}}
\tikzset{loop/.style={out=-40+#1,in=30+#1,looseness=10}}
\tikzset{thinloop/.style={out=-25+#1,in=15+#1,looseness=20}}

\tikzset{emptystep/.style={-,dotted,line cap=round,dash pattern=on 0 off 3.00000}}

\tikzset{b/.style={anchor=north,at=(#1.south)}}
\tikzset{br/.style={anchor=north west,at=(#1.south east)}}
\tikzset{bl/.style={anchor=north east,at=(#1.south west)}}
\tikzset{bw/.style={anchor=north west,at=(#1.south west)}}
\tikzset{be/.style={anchor=north east,at=(#1.south east)}}
\tikzset{a/.style={anchor=south,at=(#1.north)}}
\tikzset{ar/.style={anchor=south west,at=(#1.north east)}}
\tikzset{al/.style={anchor=south east,at=(#1.north west)}}
\tikzset{aw/.style={anchor=south west,at=(#1.north west)}}
\tikzset{ae/.style={anchor=south east,at=(#1.north east)}}
\tikzset{r/.style={anchor=west,at=(#1.east)}}
\tikzset{l/.style={anchor=east,at=(#1.west)}}
\tikzset{rn/.style={anchor=north west,at=(#1.north east)}}

\tikzset{eta/.style={very thick,->,cblue!90!black}}
\tikzset{beta/.style={very thick,->,corange!80!white!90!black}}

\tikzset{devcirc/.style={circle,draw,fill=white,inner sep=0,minimum size=4\pgflinewidth}}
\tikzset{dev/.style={postaction={decorate},decoration={
  markings,
  mark=at position .5 with \node [devcirc] {};}}
}
    
\tikzset{medium tree/.style={
    level 1/.style={sibling distance=17mm},
    level 2/.style={sibling distance=9mm},
    level 3/.style={sibling distance=5mm},
    level 4/.style={sibling distance=4mm},
  }}

\tikzset{startAt/.style={inner sep=0mm,r=#1,xshift=-1.2mm,yshift=.8mm}}

\newcommand{\arrowTriangle}[2]{
  \pgfpathmoveto{\pgfpoint{-0.01#1+#2}{.6#1}}
  \pgfpathlineto{\pgfpoint{#1+#2}{0}}
  \pgfpathlineto{\pgfpoint{-0.01#1+#2}{-.6#1}}
  \pgfusepathqfill
}

\newdimen\prearrowsize
\newdimen\arrowsize
\newdimen\temparrowsize
\newcommand{\arrowscale}{5}
\newcommand{\setarrowsize}{
  \arrowsize=0.000000001pt
  \prearrowsize=\arrowscale\pgflinewidth
  \normalizearrowsize
}
\newcommand{\normalizearrowsize}{
  \ifdim\prearrowsize>2mm
    \addtolength{\arrowsize}{2mm}
    \addtolength{\prearrowsize}{-2mm}
    \temparrowsize=0.5\prearrowsize
    \prearrowsize=\temparrowsize
  \else
  \fi

  \addtolength{\arrowsize}{\prearrowsize}
}

\pgfarrowsdeclarealias{share>}{share>}{stealth'}{stealth'}%

\pgfarrowsdeclare{red>}{red>}
{
  \setarrowsize
  \pgfarrowsleftextend{1.5\pgflinewidth} 
  \pgfarrowsrightextend{\arrowsize-0.1\arrowsize}
}{
  \setarrowsize
  \pgfsetdash{}{0pt} 
  \pgfsetroundjoin 
  \pgfsetroundcap 
  \arrowTriangle{\arrowsize}{-0.1\arrowsize}
}

\pgfarrowsdeclare{red>>}{red>>}
{
  \setarrowsize
  \pgfarrowsleftextend{1.5\pgflinewidth}
  \pgfarrowsrightextend{\arrowsize-0.1\arrowsize+.8\arrowsize}
}{
  \setarrowsize
  \pgfsetdash{}{0pt} 
  \pgfsetroundjoin 
  \pgfsetroundcap 
  \arrowTriangle{\arrowsize}{-0.1\arrowsize}
  \arrowTriangle{\arrowsize}{-0.1\arrowsize+.8\arrowsize}
}

\pgfarrowsdeclare{red>>>}{red>>>}
{
  \setarrowsize
  \pgfarrowsleftextend{1.5\pgflinewidth}
  \pgfarrowsrightextend{\arrowsize-0.1\arrowsize+1.6\arrowsize}
}{
  \setarrowsize
  \pgfsetdash{}{0pt} 
  \pgfsetroundjoin 
  \pgfsetroundcap 
  \arrowTriangle{\arrowsize}{-0.1\arrowsize}
  \arrowTriangle{\arrowsize}{-0.1\arrowsize+.8\arrowsize}
  \arrowTriangle{\arrowsize}{-0.1\arrowsize+1.6\arrowsize}
}

\tikzset{>=red>}
\pgfarrowsdeclarealias{<<}{>>}{red>>}{red>>}%
\pgfarrowsdeclarealias{<<<}{>>>}{red>>>}{red>>>}%

\tikzstyle{gyellow}=[draw=black!80,top color=white!50,bottom color=black!20]
\tikzstyle{gblue}=[draw=blue!50,top color=white,bottom color=blue!60]
\tikzstyle{gred}=[draw=red!50,top color=white,bottom color=red!60]
\tikzstyle{ggreen}=[draw=blue!80!green!90!black,top color=white,bottom color=blue!80!green!60]

\tikzstyle{roundNode}=[gyellow,thick,circle,minimum size=4mm,inner sep=0.5mm]

\definecolor{cblue}{rgb}{0,0.4,0.7}
\definecolor{clighterblue}{rgb}{0,0.6,1.0}
\colorlet{cred}{red}
\colorlet{cgreen}{green!80!black}
\colorlet{corange}{orange!70!red}
\colorlet{cpureorange}{orange}
\colorlet{cpurple}{clighterblue!50!cred}

\colorlet{clightblue}{clighterblue!50!cblue!40}
\colorlet{clightred}{cred!40}
\colorlet{clightgreen}{cgreen!80!cblue!40}
\colorlet{clightyellow}{corange!40!yellow!50}
\colorlet{clightorange}{cred!50!orange!40}
\colorlet{clightpurple}{clighterblue!50!cred!50}

\colorlet{cdarkred}{cred!70!black}
\colorlet{cdarkgreen}{cgreen!60!black}
\colorlet{cdarkblue}{cblue!60!black}

\colorlet{chighlight}{orange!50!yellow!60}

\tikzset{pgnode/.style={smallCircle,fill=white,draw=black,minimum size=1.3mm,outer sep=0.5mm}}
\tikzset{pgnodecolor/.style={pgnode,minimum size=4mm,scale=0.9}}
\tikzset{pgnodebig/.style={roundNode,gyellow,outer sep=1mm}}
\tikzset{pgrelation/.style={ultra thick,cblue!80!black,decorate,decoration={snake,amplitude=.4mm,segment length=4mm}}}

\tikzset{exi/.style={densely dotted}}
\tikzset{decweak/.style={-,green!50!black,very thick,opacity=0.7}}
\tikzset{decstrict/.style={->,orange,very thick,opacity=0.7}}

\newcommand{\ssrc}{\mit{s}}
\newcommand{\src}{\funap{\ssrc}}
\newcommand{\stgt}{\mit{t}}
\newcommand{\tgt}{\funap{\stgt}}
\newcommand{\slab}{\mit{\ell}}
\newcommand{\lab}{\funap{\slab}}

\newcommand{\isomorph}{\approx}
\newcommand{\cxtnode}{\square}
\newcommand{\strace}{\tau}
\newcommand{\cxt}[2]{\mathit{ctx}(#1,#2)}

\tikzset{interconnect/.style={dotted,cred}}
\tikzset{match/.style={cdarkgreen,very thick}}
\tikzset{jigsaw/.style={circle,draw=black,minimum size=2mm,fill=white,minimum size=3.5mm}}
\newcommand{\vin}[5]{\draw [<-,interconnect,#3] (#1.#2) to node [pos=1.5,jigsaw,#4] {{\normalfont #5}} ++(#2:4mm);}
\newcommand{\vout}[5]{\draw [->,interconnect,#3] (#1.#2) to node [pos=1.5,jigsaw,#4] {{\normalfont #5}} ++(#2:4mm);}
\newcommand{\jigsawlabel}[2]{\;%
  \begin{tikzpicture}[default,baseline=(n.base)]
    \node (n) [interconnect,jigsaw,#2,fill=white] {{\normalfont #1}};
  \end{tikzpicture}\,%
}

\colorlet{myblue}{blue!80!black}
\colorlet{mygreen}{cdarkgreen}
\colorlet{myred}{cred}
\colorlet{myorange}{corange}
\colorlet{mypurple}{blue!40!cred}

\tikzset{smalljigsaw/.style={rectangle,rounded corners=3mm,inner sep=1mm}}
\newcommand{\noimplict}[1]{
  \begin{scope}[line width=1.1mm,myred,opacity=0.3]
    \draw ($(#1)+(-4mm,-4mm)$) to ++(8mm,8mm);
    \draw ($(#1)+(-4mm,4mm)$) to ++(8mm,-8mm);
  \end{scope}
}

\newcommand{\tikzstep}[1]{
        \draw [->,ultra thick,n/.style={graphNode}] (#1mm,0) to ++(5mm,0mm);
}

\tikzset{graphNode/.style={circle,draw=black,inner sep=.5mm,outer sep=1mm}}

\begin{document}

\title{Patch Graph Rewriting (Extended Version)\protect\thanks{ 
The present version extends the submission accepted to ICGT20 with an appendix.
The final authenticated publication is available online at \url{https://doi.org/10.1007/978-3-030-51372-6_8}.
}}

\author{Roy Overbeek\,\textsuperscript{\faEnvelopeO} \and J\"{o}rg Endrullis}
\authorrunning{Roy Overbeek and J\"{o}rg Endrullis}

\institute{Vrije Universiteit Amsterdam, Amsterdam, The Netherlands \\
  \email{\{r.overbeek, j.endrullis\}@vu.nl}
}
\maketitle

\begin{abstract}
  \noindent 
  The basic principle of graph rewriting is the stepwise replacement of subgraphs inside a host graph. A challenge in such replacement steps is the treatment of the \emph{patch graph}, consisting of those edges of the host graph that touch the subgraph, but are not part of it.
  
  We introduce \emph{patch graph rewriting}, a visual graph rewriting language with precise formal semantics. The language has rich expressive power in two ways. First, rewrite rules can flexibly constrain the permitted shapes of patches touching matching subgraphs. Second, rules can freely transform patches. We highlight the framework's distinguishing features by comparing it against existing approaches.

\end{abstract}

\section{Introduction}

When matching a graph pattern $P$ inside a host graph $G$, $G$ can be partitioned into (i) a \emph{match} $M$, a subgraph of $G$ isomorphic to the pattern $P$; (ii) a \emph{context} $C$, the largest subgraph of $G$ disjoint from $M$; and (iii) a \emph{patch} $J$, the graph consisting of the edges that are neither in $M$ nor in $C$.
So the patch consists of edges that are either (a) between $M$ and $C$, in either direction, or (b) between vertices of $M$ not captured by the pattern $P$.
For example, if $P$ and $G$ are respectively
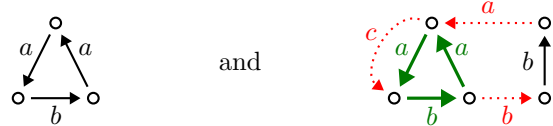
\begin{center}
\begin{tikzpicture}[default,node distance=10mm,n/.style={graphNode}]
\begin{scope}
\node (3)[n] {};
\node (4)[n] [right of=3] {};
\node (5)[n] at ($(3)!.5!(4) + (0mm,10mm)$) {};
\draw [->] (3) to node [below] {$b$} (4);
\draw [->] (4) to node [above right] {$a$} (5);
\draw [->] (5) to node [above left] {$a$} (3);
\end{scope}
\node at (29.5mm,5mm) {and};
\begin{scope}[xshift=50mm]
    \node (3)[n] {};
    \node (4)[n] [right of=3] {};
    \node (5)[n] at ($(3)!.5!(4) + (0mm,10mm)$) {};
    
    \node(6)[n] [right of=4] {};
    \node(7)[n] [above of=6] {};
    
    \draw [->, match] (3) to node [below] {$b$} (4);
    \draw [->, match] (4) to node [above right] {$a$} (5);
    \draw [->, match] (5) to node [above left] {$a$} (3);
    \draw [->, myred, interconnect] (4) to node [below] {$b$} (6);
    \draw [->, myred, interconnect] (7) to node [above] {$a$} (5);
    \draw [->] (6) to node [left] {$b$} (7);
    \draw [->, myred, interconnect] (5) to[bend right=80] node [above left] {$c$} (3);
\end{scope}
\end{tikzpicture}
\end{center}
then the thick green subgraph is the (only) match $M$ of $P$ in $G$. The black subgraph of $G$ is the context $C$, and the dotted red subgraph is the patch $J$. Metaphorically, patch~$J$ patches match $M$ and context $C$ together.

In graph rewriting, subgraphs of some host graph are stepwise replaced by other subgraphs. A requirement for such replacements is that they are properly re-embedded in the host graph.
We contend that the patch is the most distinctive and  interesting aspect of graph rewriting. This is because its shape is generally unpredictable, making it challenging to specify what constitutes a proper re-embedding of a subgraph replacement. This contrasts strongly with the situation for string and term rewriting, in which the embeddings of substrings and subterms are highly regular.

Most existing approaches to graph rewriting are rather uniform and coarse-grained in their treatment of the patch. For instance, suppose that we wish to delete the match $M$ from $G$. What should happen to the edges of patch $J$, which would be left ``dangling'' by such a removal? The popular double-pushout (DPO)~\cite{ehrig1973graph} approach to graph rewriting  conservatively dictates that the application is not allowed in the first place: nodes connected to the patch \emph{must} be preserved by the rewrite step, and the patch shall remain connected as before. The single-pushout (SPO)~\cite{lowe1993algebraic} variant, by contrast, permissively answers that such a deletion is always possible. As a side-effect, however, any resulting dangling patch edges are discarded.

In this paper, we introduce the \emph{patch graph rewriting} (\emph{PGR}) language. It has the following features:
\begin{itemize}
    \item \emph{Pluriform, fine-grained control over patches}. Rules themselves encode which kinds of patches are allowed around matches, as well as how they should be transformed for the re-embedding, using a \emph{unified} notation. Thus, these policies are distinctly not decided on the level of the framework.
    \item \emph{An intuitive visual language}. Despite their expressive power and formal semantics, patch rewrite rules admit a visual representation that we believe to be highly intuitive.
    \item \emph{Lightweight formal semantics}. The formal details of PGR are based on elementary set and graph theory, and therefore accessible to a wide audience. In particular, an understanding of category theory is not required to understand these details, unlike for many dominant approaches in graph~rewriting.
\end{itemize}

The remainder of our paper is structured as follows. To fix ideas and emphasize the visual language of PGR, we first provide an intuitive exposition in Section~\ref{sec:visual}, and then follow with a formal introduction in Section~\ref{sec:graph:rewriting}. We show the usefulness of PGR by modeling wait-for graphs and deadlock detection in Section~\ref{sec:wait:for:graphs}.
We compare PGR to other approaches in Section~\ref{sec:comparison}. In Section~\ref{sec:conclusion}, finally, we mention some future research directions for PGR.

\section{Intuitive Semantics}

\input{visual.tex}\label{sec:visual}

\section{Formal Semantics}\label{sec:graph:rewriting}

\begin{notation}[Preliminaries] For functions $f : A_f \to B_f$ and $g : A_g \to B_g$ 
with disjoint domains (but possibly overlapping codomains),
we write ${f} \cup {g}$ for the function $(f \, \cup \, g) : (A_f \, \cup \, A_g) \to (B_f \, \cup \, B_g)$ given by the union of $f$ and $g$'s underlying graphs. If typing permits, we generalize functions $f$ to tuples $(x,y)$ and sets $S$ in the obvious way, i.e., $f ((x,y)) = (f(x), f(y))$ and $f(S) = \{\, f(x) \mid x \in S \,\}$.
\end{notation}

We define directed, edge-labeled multigraphs in the standard way.

\begin{definition}[Graph]
  A \emph{graph} $G = (V,E,\ssrc,\stgt,\slab)$ with edge labels from $L$ consists of
  a finite set of \emph{vertices} (or \emph{nodes}) $V$, a finite set of edges $E$, a \emph{source map} $\ssrc: E \to V$, a \emph{target map} $\stgt: E \to V$, and a \emph{labeling} $\slab: E \to L$.
  For $e \in E$, we say that $\src{e}$, $\tgt{e}$ and $\lab{e}$
  are the \emph{source}, \emph{target} and \emph{label} of $e$, respectively.
  
    For convenience, we will write $x \xrightarrow{\alpha} y \in E$ to denote an edge $e \in E$ such that $\ssrc(e) = x$, $\stgt(e) = y$ and $\slab(e) = \alpha$.
\end{definition}

We depict graphs as usual. An explanation may be found in Appendix~\ref{sec:depicting:graphs}. A discussion on how to encode vertex labels as edge labels is found in Appendix~\ref{sec:vertex:labels}.

\begin{definition}[Basic Graph Notions]
  We define the following basic graph notions:
  \begin{enumerate}[label=(\roman*)]
    \item 
      An \emph{unlabeled graph} $G = (V,E,\ssrc,\stgt)$ is a graph $(V,E,\ssrc,\stgt,\slab)$ over a singleton label set. In this case we suppress the edge labels.
    \item 
      A graph is \emph{simple} if for all $e,e' \in E$, $\ssrc(e) = \ssrc(e')$, $\stgt(e) = \stgt(e')$ and $\slab(e) = \slab(e')$ together imply $e = e'$.
    \item 
      We say that graphs $G$ and $H$ are \emph{disjoint} if $V_G \cap V_H = \emptyset = E_G \cap E_H$. 
    \item 
      For disjoint edge sets $E \cap E' = \emptyset$, we define the \emph{graph union} 
      as follows:
      \begin{align*}
         (V, E, \ssrc, \stgt, \slab) \cup (V', E', \ssrc', \stgt', \slab')
         \;\;=\;\;
         (V \cup V', E \cup E', \ssrc \cup \ssrc', \stgt \cup \stgt', \slab \cup \slab') \;.
      \end{align*}
  \end{enumerate}
\end{definition}

To rename vertices and edges of a graph, we introduce ``graph renamings''.
A renaming is a graph isomorphism, where the domain of the renaming is allowed to be a 
superset of the set of vertices/edges of the graph.
In this way, the same renaming can be applied to various graphs
with different vertex and edge sets.

\begin{definition}[Graph Renaming]
  A \emph{graph renaming} $\phi$ for a graph $G$
  consists of two bijective functions $\phi_V : V_1 \to V_2$ and
  $\phi_E : E_1 \to E_2$ such that $V_G \subseteq V_1$ and $E_G \subseteq E_1$. 
  
  The $\phi$-renaming of $G$, denoted $\phi(G)$,
  is the graph $(V,E,\ssrc,\stgt,\slab)$ defined by
  \begin{align*}
    V &= \phi_V(V_G) &
    \src{\phi_E(e)} &= \phi_V(\ssrc_G(e)) & 
    \lab{\phi_E(e)} &= \slab_G(e)
    \\
    E &= \phi_E(E_G) &
    \tgt{\phi_E(e)} &= \phi_V(\stgt_G(e))
  \end{align*}
  for every $e \in E_G$.
\end{definition}

\begin{definition}[Graph Isomorphism]
  Graphs $G$ and $H$ are \emph{isomorphic}, denoted $G \isomorph H$,
  if there is a graph renaming $\phi$ such that $H = \phi(G)$.
\end{definition}

Let $L$ be a finite nonempty set of labels.
In the sequel, we tacitly assume that all graphs have labels from $L$.

As motivated by the preceding sections, we allow to compose a context graph $C$ and a match graph $M$ by a ``patch'' $J$ that may  add edges between the nodes of $C$ and $M$, as well as between the nodes of $M$.

\begin{definition}[Patch]
  Let $C$ and $M$ be disjoint graphs.
  A \emph{patch} for $C$ and $M$ is a graph $J$ 
  such that $E_J \cap (E_C \cup E_M) = \emptyset$ and $V_J = \ssrc(E_J) \cup \stgt(E_J)$, and
  \begin{align*}
    (\ssrc(e),\stgt(e)) \in (V_C \times V_M) \cup (V_M \times V_C) \cup (V_M \times V_M)
  \end{align*}
  for every edge $e \in E_J$.
  In this relation mediated by $J$, we call $C$ the \emph{context graph} and $M$ the \emph{match graph}.
\end{definition}

\begin{definition}[Patch Composition]
  Let $J$ be a patch for a context graph $C$ and a match graph $M$. The \emph{patch composition} of $C$ and $M$ 
  through \emph{patch} $J$, denoted by $C \cdot_J M$, is the graph union $C \cup J \cup M$.
\end{definition}

\begin{example}
  Consider the following graphs $C$, $M$ and $G$, respectively:
  
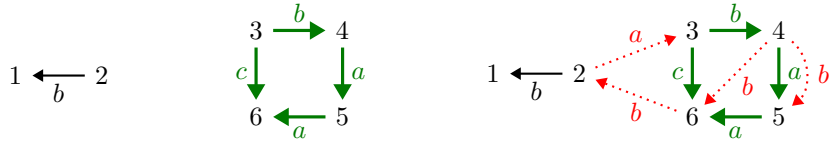
\begin{center}
   \begin{tikzpicture}[default,node distance=11.5mm,baseline=-6.5ex]
        \begin{scope}[yshift=-6mm]
          \node (1) {$1$};
          \node (2) [right of=1] {$2$};
          \draw [->] (2) to node [below] {$b$} (1);
        \end{scope}
        \begin{scope}[xshift=32mm]
          \node (3) {$3$};
          \node (4) [right of=3] {$4$};
          \node (5) [below of=4] {$5$};
          \node (6) [below of=3] {$6$};
          \draw [->, match] (3) to node [above] {$b$} (4);
          \draw [->, match] (4) to node [right] {$a$} (5);
          \draw [->, match] (5) to node [below] {$a$} (6);
          \draw [->, match] (3)
          to node [left] {$c$}
          (6);
        \end{scope}
        \begin{scope}[xshift=90mm]
          \node (3) {$3$};
          \node (4) [right of=3] {$4$};
          \node (5) [below of=4] {$5$};
          \node (6) [below of=3] {$6$};
          \draw [->, match] (3) to node [above] {$b$} (4);
          \draw [->, match] (4) to node [right] {$a$} (5);
          \draw [->, match] (5) to node [below] {$a$} (6);
          \draw [->, match] (3)
          to node [left] {$c$}
          (6);
        \node (2) at ($(3)!.5!(6) + (-15mm,0mm)$) {$2$};
        \node (1) [left of=2] {$1$};
        \draw [->] (2) to node [below] {$b$} (1);
        \draw [->, interconnect] (2) to node [above] {$a$} (3);
        \draw [->, interconnect] (6) to node [below] {$b$} (2);
        \draw [->, interconnect] (4) to node [below right] {$b$} (6);
        \draw [->, interconnect] (4) to[bend left=50] node [right] {$b$} (5);
\end{scope}
\end{tikzpicture}
\end{center}
  The composition of $C$ and $M$
  through patch 
  $J = \set{ 
    2 \stackrel{a}{\to} 3,\; 
    6 \stackrel{b}{\to} 2,\; 
    4 \stackrel{b}{\to} 5,\;
    4 \stackrel{b}{\to} 6
  }$ is $G$,
  in which $C$ functions as the context graph and $M$ functions as the match graph (w.r.t.\ $J$).
\end{example}

Before we consider the formal definition of rewriting,  let us discuss the basic principle and motivate some of the design choices. As a first approximation, a graph rewrite rule $L \to R$ is a pair of graphs that behave like patterns. Since the edge and vertex identities in such rules are arbitrary (not to be confused with the edge labeling), we close the rule under isomorphism. 
The rule should also be applicable in contexts in which a patch connects a context and the pattern of the rule. The rule $L \to R$ thus give rise to rewrite steps of the form
  $C \cdot_J L' \to C \cdot_{J'} R'$ 
for graphs $C$, valid patches $J, J'$ and graphs $L' \approx L$ and $R' \approx R$.

Additionally, we would like to exert control over the shape of patches in two ways. 
A graph rewriting rule should enable one to
(a) constrain the choices for the patch~$J$, and
(b) define the patch $J'$ in terms of $J$.
For these purposes, we introduce the concepts of a patch type and a scheme.

\begin{definition}[Patch Type]\label{def:patch:type}
  A \emph{patch type}~$T$ for a graph $G$ is an unlabeled patch for $G$ and the trivial graph 
  with node set $\{\, \cxtnode \,\}$. Here, the trivial graph functions as the context graph.
  
  Let $J$ be a patch for a context graph $C$ and match graph $M$, and $T$ a patch type for $M$. A patch edge $(j_s \xrightarrow{\alpha} j_t) \in E_J$ ($\alpha$ any label) \emph{adheres to} a patch type edge $(t_s \to t_t) \in E_T$ if the following conditions hold:
  \begin{align*}
    j_s \in V_C \, &\Rightarrow \, t_s = \cxtnode &
    j_s \in V_M \, &\Rightarrow \, j_s = t_s \\ 
    j_t \in V_C \, &\Rightarrow \, t_t = \cxtnode &
    j_t \in V_M \, &\Rightarrow \, j_t = t_t
  \end{align*}
  The patch $J$ \emph{adheres to} patch type $T$ if there exists an \emph{adherence map} from $J$ to~$T$, i.e., a function $f: E_J \to E_T$ such that $e$ adheres to $f(e)$ for every $e \in E_J$.
\end{definition}

The restriction to unlabeled patch type edges is motivated purely by simplicity. We intend to relax the definition in future work.

\newcommand{\sadheres}[2]{\mathit{adh}_{#1}(#2)}

\begin{proposition}[Unique Adherences]\label{prop:simple:adheres:to:one}
Let the patch type $T$ be a simple graph. If a patch $J$ adheres to $T$, then the witnessing adherence map is unique.
\end{proposition}

Intuitively, we use patch types to annotate the patterns of a rewrite rule. The result we call a scheme.

\begin{definition}[Scheme]
  A \emph{scheme} is a pair $(P,T)$  consisting of a graph $P$, called a \emph{pattern}, and a patch type $T$ for~$P$.
\end{definition}

\vspace{-2ex}
\noindent
\begin{minipage}{0.65\textwidth}%
\noindent
\begin{example}[Depicting Schemes]
  We extend the representation for graphs to schemes $(P,T)$ as shown on the right.
  The pattern $P$ consists of the solid labeled arrows,
  and the patch type $T$ consists of the dotted arrows.
  For dotted arrows without a source (or target),
  the source (or target) is implicitly the context graph node $\cxtnode$. So 
%
  $T$ consists of the edges 
  $\set{ 
    \cxtnode \to 1,\; 3 \to \cxtnode,\;
    \cxtnode \to 4,\; 4 \to \cxtnode,\;
    1 \to 3
  }$.
\end{example}
\end{minipage}\hfill%
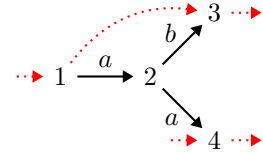
\begin{minipage}{0.32\textwidth}
  \begin{tikzpicture}[default,node distance=12mm,]
    \node (1) [] {$1$};
    \node (2) [right of=1] {$2$};
    \node (3) [above right of=2] {$3$};
    \node (4) [below right of=2] {$4$};

    \draw [->] (1) to node [above] {$a$} (2);
    \draw [->] (2) to node [above left] {$b$} (3);
    \draw [->] (2) to node [below left] {$a$} (4);
    
    \vin{1}{-180}{red}{draw=none,minimum size=0}{}
    \vin{4}{-180}{red}{draw=none,minimum size=0}{}
    \vout{3}{0}{red}{draw=none,minimum size=0}{}
    \vout{4}{0}{red}{draw=none,minimum size=0}{}
    \draw [->,interconnect] (1) to[bend left=30] (3);
  \end{tikzpicture}
\end{minipage}
\bigskip

We are now ready to define a graph rewrite rule as a relation between schemes $(P_L, T_L)$ and $(P_R, T_R)$. We equip the rewrite rule with a ``trace function'' $\strace$ that relates edges in $T_R$ back to edges in $T_L$, allowing us to interpret $T_R$ as a transformation of $T_L$, in which patch edges may be freely moved, deleted, duplicated and inverted. For this we require the following constraint: if a patch type edge $e \in E_{T_R}$ connects to the context, the corresponding edge $\tau(e) \in E_{T_L}$ must also connect to the context. Without this constraint, it would not be clear how to interpret $e$'s connection to the context.

\begin{definition}[Quasi Patch Graph Rewrite Rule]\label{def:graph:rewrite:rule}
  A \emph{quasi patch graph rewrite rule} $L \xrightarrow{\strace} R$ is a pair of schemes $L = (P_L, T_L)$ and $R = (P_R, T_R)$, equipped with a \emph{trace function} $\strace : E_{T_R} \to E_{T_L}$ that satisfies
  $
    \cxtnode \in \{ \ssrc(e), \stgt(e) \} \, \Longrightarrow \, \cxtnode \in \{ \ssrc(\strace(e)), \stgt(\strace(e)) \}
  $
  for all $e \in E_{T_R}$.
\end{definition}

We normally require the left patch type $T_L$ to be simple, so that the edges of $T_L$-adherent patches $J$ adhere to a single edge in $T_L$~(Proposition~\ref{prop:simple:adheres:to:one}). As we shall see, this allows us to define a graph rewrite relation in which matches of a rule produce a unique result (modulo~$\isomorph$).

\begin{definition}[Patch Graph Rewrite Rule]
A \emph{patch graph rewrite rule} is a quasi patch graph rewrite rule $(P_L, T_L) \xrightarrow{\tau} (P_R, T_R)$ in which $T_L$ is simple.
\end{definition}

Since we restrict attention to unlabeled patch type graphs in this paper, we will use the opportunity to visualize the trace function $\strace$ by means of labels on patch type edges.

\begin{example}[Depicting Rules]\label{ex:rewrite:rule}
A depiction of a valid rewrite rule is:
  
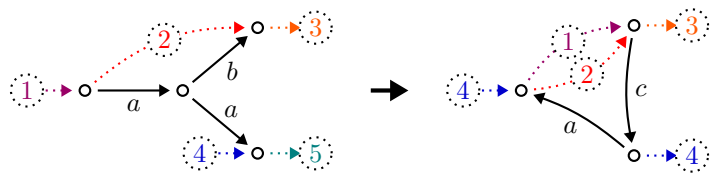
\begin{center}
  \begin{tikzpicture}[default,node distance=18mm,n/.style={graphNode}]
    \begin{scope}[xshift=0mm]
        \node (2) [n] {};
        \node (1) [n] at ($(2)+(180:13mm)$) {};
        \node (3) [n] at ($(2)+(40:13mm)$) {};
        \node (4) [n] at ($(2)+(-40:13mm)$) {};
    
        \draw [->] (1) to node [below] {$a$} (2);
        \draw [->] (2) to node [below right] {$b$} (3);
        \draw [->] (2) to node [above right] {$a$} (4);
        \draw [->,interconnect,myred] (1) to[bend left=20] node [jigsaw] {$2$} (3);
        
        \vin{1}{180}{mypurple}{}{$1$}
        \vout{3}{0}{myorange}{}{$3$}
        \vin{4}{180}{myblue}{}{$4$}
        \vout{4}{0}{teal}{}{$5$}
    \end{scope}
    \tikzstep{25}
    \begin{scope}[xshift=55mm]
        \node (2) [] {};
        \node (1) [n] at ($(2)+(180:10mm)$) {};
        \node (3) [n] at ($(2)+(60:10mm)$) {};
        \node (4) [n] at ($(2)+(-60:10mm)$) {};
    
        \draw [->,interconnect] (1) to[bend left=-25] node [jigsaw] {$2$}  (3);
        \draw [->,interconnect,mypurple] (1) to[bend left=25] node [jigsaw] {$1$}  (3);
        \draw [->] (3) to[bend left=-10] node [right] {$c$} (4);
        \draw [->] (4) to[bend left=-10] node [below left] {$a$} (1);
        
        \vin{1}{180}{myblue}{}{$4$}
        \vout{4}{0}{myblue}{}{$4$}
        \vout{3}{0}{myorange}{}{$3$}
    \end{scope}
  \end{tikzpicture}
  \end{center}
  The trace function $\strace$ is visualized by means of labels on the type edges: $\tau$ maps 
  type edges with label $n$ on the right-hand side to the type edge with label $n$ on the left-hand side. \emph{Throughout the paper, colors are merely used as a supplementary visual aid.}
  (An application example will be given in Example~\ref{ex:rewrite}.)
\end{example}

\begin{definition}[Rule Isomorphism]
  (Quasi) rewrite rules $L_1 \xrightarrow{\strace_1} R_1$ and $L_2 \xrightarrow{\strace_2} R_2$ are \emph{isomorphic}, 
  denoted ${L_1 \xrightarrow{\strace_1} R_1} \isomorph {L_2 \xrightarrow{\strace_2} R_2}$,
  if there exists a graph renaming $\phi$ such that
     $\phi_V(\cxtnode) = \cxtnode$,
     $\phi((L_1, R_1)) = (L_2, R_2)$, and 
    $\phi_E \circ \strace_1 = \strace_2 \circ \phi_E$.
\end{definition}

\newcommand{\ssystem}{\mathcal{R}}
\begin{definition}[Patch Graph Rewrite System]
  A \emph{(quasi) patch graph rewrite system} (\emph{PGR}) $\ssystem$ is a set of (quasi) rewrite rules. For $\ssystem$ we define the isomorphism closure class
  $\ssystem^\isomorph = \set{ y \mid x \in \ssystem,\;y \isomorph x}$.
\end{definition}

For a patch $J$, patch type $T$ and adherence map $h: E_J \to E_T$, we define
\begin{align*}
  \cxt{e}{h} = 
    \begin{cases}
      \set{\src{e}} &\text{if $\src{h(e)} = \cxtnode$}\\
      \set{\tgt{e}} &\text{if $\tgt{h(e)} = \cxtnode$}\\
      \emptyset &\text{otherwise}
    \end{cases}
\end{align*}
for every $e \in E_J$.
So $\cxt{e}{h}$ contains the context node involved in the edge $e$,
or is $\emptyset$ if the edge does not involve the context.

\begin{definition}[Patch Graph Rewriting]\label{def:graph:rewriting}
  A (quasi) patch graph rewrite system $\ssystem$ induces a rewrite relation~$\to_\ssystem$ on the set of graphs as follows:\\
    $C \cdot_J P_L \; \to_\ssystem \; C \cdot_{J'} P_R$
  if
  \begin{enumerate}[label=(\roman*)]
    \item\label{def:graph:rewriting:one} $(P_L, T_L) \xrightarrow{\strace} (P_R, T_R) \in \ssystem^\isomorph$,
    \item\label{def:graph:rewriting:two} $h_L : E_J \to E_{T_L}$ is an adherence map from patch $J$ to patch type $T_L$,
    \item\label{def:graph:rewriting:three} $h_R : E_{J'} \to E_{T_R}$ is an adherence map from patch $J'$ to patch type $T_R$, and
    \item\label{def:graph:rewriting:four}
      for every $t \in E_{T_R}$ there exists a 
      bijection $\sigma : h_R^{-1}(t) \to h_L^{-1}(\strace(t))$ such that
        $\slab_R(e) = \slab_L(\sigma(e))$ and
        $\cxt{e}{h_R} \subseteq \cxt{\sigma(e)}{h_L}$
    for every $e \in h_R^{-1}(t)$. 
  \end{enumerate}
  For such a rewrite step, we say that the graph $C \cdot_J P_L$ contains the \emph{redex} $P_L$.
\end{definition}

    \begin{example}[Application Example]\label{ex:rewrite}
  The rule given in Example~\ref{ex:rewrite:rule} gives rise to the following rewrite step:
  
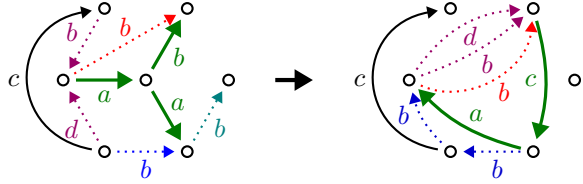
\begin{center}\vspace{1ex}
  \begin{tikzpicture}[default,node distance=11mm,n/.style={graphNode}]
    \begin{scope}[xshift=0mm]
        \node (2) [n] {};
        \node (1) [n] at ($(2)+(180:11mm)$) {};
        \node (3) [n] at ($(2)+(60:11mm)$) {};
        \node (4) [n] at ($(2)+(-60:11mm)$) {};
        \node (34) [n] at ($(2)+(0:11mm)$) {};
        \node (14) [n] at ($(2)+(-120:11mm)$) {};
        \node (13) [n] at ($(2)+(120:11mm)$) {};
    
        \begin{scope}[match]
        \draw [->] (1) to node [below] {$a$} (2);
        \draw [->] (2) to node [below right] {$b$} (3);
        \draw [->] (2) to node [above right] {$a$} (4);
        \end{scope}
        \begin{scope}[cred,interconnect]
        \draw [->, red] (1) to node [above] {$b$} (3);
        \draw [->, mypurple] (14) to node [below left] {$d$} (1);
        \draw [->, blue] (14) to node [below] {$b$} (4);
        \draw [->, teal] (4) to node [below right] {$b$} (34);
        \draw [->, mypurple] (13) to node [above left] {$b$} (1);
        \end{scope}
        \begin{scope}[black]
        \draw [->] (14) to[bend left=80,looseness=1.6] node [left] {$c$} (13);
        \end{scope}
    \end{scope}
    \tikzstep{17.2}
    \begin{scope}[xshift=46mm]
        \node (2) {};
        \node (1) [n] at ($(2)+(180:11mm)$) {};
        \node (3) [n] at ($(2)+(60:11mm)$) {};
        \node (4) [n] at ($(2)+(-60:11mm)$) {};
        \node (34) [n] at ($(2)+(0:11mm)$) {};
        \node (14) [n] at ($(2)+(-120:11mm)$) {};
        \node (13) [n] at ($(2)+(120:11mm)$) {};
    
        \begin{scope}[match]
        \draw [->] (3) to[bend left=15] node [left] {$c$} (4);
        \draw [->] (4) to[bend left=15] node [above right] {$a$} (1);
        \end{scope}
        \begin{scope}[interconnect]
        \draw [->, myred] (1) to[bend left=-55] node [below right] {$b$} (3);
        \draw [->, mypurple] (1) to[bend left=-10] node [below right] {$b$} (3);
        \draw [->, mypurple] (1) to[bend left=25] node [below right,inner sep=0] {$d$} (3);
        \draw [->, myblue] (14) to[bend left=15] node [left,pos=.6] {$b$} (1);
        \draw [->, myblue] (4) to[bend left=0] node [below,pos=0.4] {$b$} (14);
        \end{scope}
        \begin{scope}[black]
        \draw [->] (14) to[bend left=80,looseness=1.6] node [left] {$c$} (13);
        \end{scope}
    \end{scope}
  \end{tikzpicture}\vspace{-1ex}
  \end{center}
  In the graph on the left we have highlighted the match (thick green) and the patch (dotted). We have additionally indicated the adherence map of the patch edges by reusing the colors of the rule definition.
\end{example}
We refer to Section~\ref{sec:visual} for many examples of simple rewrite rules and
rewrite steps, and to Appendix~\ref{sec:example:rules} for rewrite rules demonstrating standard graph operations such as merging, copying and splitting nodes. A graph rewrite system modeling wait-for graphs will be given in Section~\ref{sec:wait:for:graphs}. 

\begin{remark}[Finding Redexes]\label{rem:check:redex}
    Checking for the presence of a redex is simple.
    A graph $G$ contains a redex with respect to rule $(P_L, T_L) \xrightarrow{\strace} (P_R, T_R) \in \ssystem$ 
    if and only if
    \begin{enumerate}
      \item\label{rem:check:redex:subgraph}
        there exists a subgraph $M$ of $G$ isomorphic to $P_L$, and
      \item\label{rem:check:redex:adherence}
        every edge $e \notin E_M$ incident to a $v \in V_M$ adheres to an edge in $T_L$.
    \end{enumerate}
\end{remark}

Definition~\ref{def:graph:rewriting} can be understood in more operational terms as follows.

\begin{lemma}[Constructing $J'$]\label{lem:constructing:j'}
  If conditions \ref{def:graph:rewriting:one} and \ref{def:graph:rewriting:two} of Definition~\ref{def:graph:rewriting} are satisfied (fixing some adherence map $h_L$), the patch $J'$ and adherence map $h_R$ that satisfy conditions \ref{def:graph:rewriting:three} and \ref{def:graph:rewriting:four} are uniquely determined up to isomorphism.
  The patch $J'$ can be constructed using the following procedure.

For every type edge $t = (t_s \to t_t) \in E_{T_R}$, consider every patch edge $j = (j_s \xrightarrow{\alpha} j_t) \in E_J$ for which $h_L(j) = \strace(t) = (t_s^\tau \to t_t^\tau) \in E_{T_L}$. There are five exclusive cases:
    \begin{enumerate}
        \item If $\cxtnode \notin \{ t_s, t_t \}$, add a new edge $t_s \xrightarrow{\alpha} t_t$ to $J'$.
        \item If $t_s = t_s^\tau = \cxtnode$, add a new edge $j_s \xrightarrow{\alpha} t_t$ to $J'$.
        \item If $t_t = t_t^\tau = \cxtnode$, add a new edge $t_s \xrightarrow{\alpha} j_t$ to $J'$.
        \item If $t_s = t_t^\tau = \cxtnode$, add a new edge $j_t \xrightarrow{\alpha} t_t$ to $J'$.
        \item If $t_t = t_s^\tau = \cxtnode$, add a new edge $t_s \xrightarrow{\alpha} j_s$ to $J'$.
  \end{enumerate}
  Here, the ``new'' edge $j'$ is an edge not in $C$, $P_R$ or the intermediate construction of $J'$. The adherence map $h_R$ is defined such that $h_R(j') = t$ for each of the considered $j'$ and $t$.
\end{lemma}

Non-quasi rules have the following desirable property.

\begin{proposition}[Rule Determinism]
Let $G = C \cdot_J P_L$. If a rule $(P_L, T_L) \xrightarrow{\strace} (P_R, T_R) \in \ssystem^\isomorph$ derives both $G \to_\ssystem C \cdot_{J'} P_R = G'$
and
$
G \to_\ssystem C \cdot_{J''} P_R = G''
$,
then $G' \isomorph G''$.
\end{proposition}
\begin{proof}
This is a direct consequence of Proposition~\ref{prop:simple:adheres:to:one} and Lemma~\ref{lem:constructing:j'}. \qed
\end{proof}
 
In contrast to (non-quasi) graph rewrite rules, quasi rules are not generally deterministic. For instance, consider the quasi rewrite rule
  
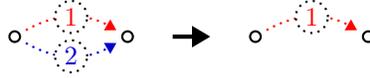
\begin{center}
  \begin{tikzpicture}[default,node distance=15mm,n/.style={graphNode}]
    \begin{scope}[xshift=0mm]
        \node (1) [n] {};
        \node (2) [n,right of=1] {};
        
        \draw [->, interconnect] (1) to [bend left=30] node [jigsaw] {$1$} (2);
        \draw [->, interconnect,myblue] (1) to [bend left=-30] node [jigsaw] {$2$} (2);
    \end{scope}
    \tikzstep{21}
    \begin{scope}[xshift=32mm]
        \node (1) [n] {};
        \node (2) [n,right of=1] {};
        
        \draw [->, interconnect] (1) to [bend left=30] node [jigsaw] {$1$} (2);
    \end{scope}
  \end{tikzpicture}
  \end{center}
  which can match any graph $G$ consisting of two nodes $x$ and $y$ and $n$ edges from $x$ to $y$. For each $e \in E_G$, the left adherence map $h_L$ can either map $e$ to the patch type edge 
  labeled with $1$, or to the type edge labeled with $2$. Thus, $2^n$ choices for $h_L$ are 
  possible, and each choice causes a different subset of $J$ to be deleted in a single rewrite step.

\begin{notation}[Shorthand Notation]\label{notation:shorthand}
Given a pattern $P$, we often want to allow for any patch edges between the nodes of a subset $N \subseteq V_P$ as well as the context node $\cxtnode$.
  In the notation we have discussed so far, we would then need to draw the complete patch type graph induced by $N \cup \{\, \cxtnode \,\}$ (minus the loop on $\cxtnode$), which consists of $(|N|+1)^2 - 1$ patch type edges.
  
  To avoid spaghetti-like figures, we extend the visual presentation of schemes by allowing each node to be annotated with a set of names (written without set brackets). We say that a node \emph{has name} $x$ if $x$ is in the set of names of this node. So a node can have $0$ or more names.
  The name annotations are then shorthand for the following patch type edges:
  \begin{enumerate}[label=(\roman*)]
    \item
      For every node $n$ and name $x$ of $n$,
      the node $n$ has the two patch type edges
      \begin{tikzpicture}[default,node distance=25mm,n/.style={graphNode},baseline=(1.base)]
        \begin{scope}[xshift=0mm]
            \node (1) {$n$}; 
            \vin{1}{180}{myred}{smalljigsaw,pos=2.5}{$(\cxtnode,x)$}
            \vout{1}{0}{myred}{smalljigsaw,pos=2.5}{$(x,\cxtnode)$}
        \end{scope}
      \end{tikzpicture} from and to the context.
    \item 
      For every pair of nodes $n,m$ and every name $x$ of $n$ and $y$ of $m$,
      there is implicitly the patch type edge
      \begin{tikzpicture}[default,node distance=25mm,n/.style={graphNode},baseline=(1.base)]
        \begin{scope}[xshift=0mm]
            \node (1) [] {$n$}; 
            \node (2) [right of=1] {$m$}; 
            \draw [->,interconnect] (1) to node [jigsaw,smalljigsaw] {$(x,y)$} (2);
        \end{scope}
      \end{tikzpicture} from  from $n$ to $m$.
      Here $n$ and $m$ can be the same node, and $x$ can be equal to $y$.
  \end{enumerate}
  Observe that rules are non-quasi iff every node in the left-hand side has at most one name. We therefore require that distinct nodes do not share names.
  
  As an example, a rule for merging two nodes can be written as
  \begin{align}
  \begin{tikzpicture}[default,node distance=10mm,n/.style={graphNode},baseline=0ex]
    \begin{scope}[xshift=0mm]
        \node (1) [n] {}; 
          \node [anchor=center,at=(1.south),inner sep=0,yshift=-1.1mm] {$x$};
        \node (2) [n,right of=1] {}; 
          \node [anchor=center,at=(2.south),inner sep=0,yshift=-1.1mm] {$y$};
        \draw [->] (1) to[bend right=0] node [above] {$a$} (2);
    \end{scope}
    \tikzstep{16}
    \begin{scope}[xshift=27mm]
        \node (1) [n] {};
          \node [anchor=center,at=(1.south),inner sep=0,outer sep=0,yshift=-1.1mm] {$x,y$};
    \end{scope}
  \end{tikzpicture}
  \label{rule:merge:abbreviated}
  \end{align}  
  which is shorthand for
  
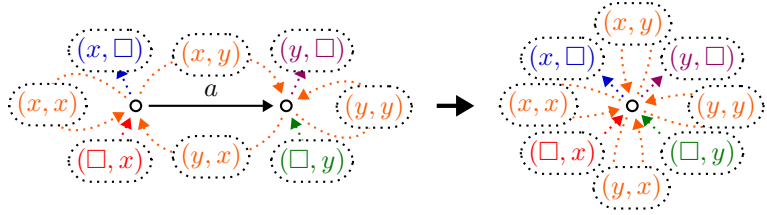
\begin{align*}
  \begin{tikzpicture}[default,node distance=20mm,n/.style={graphNode},baseline=0ex]
    \begin{scope}[xshift=0mm]
        \node (1) [n] {};
        \node (2) [n,right of=1] {};
        \draw [->] (1) to[bend right=0] node [above] {$a$} (2);
        \vin{1}{-115}{myred}{smalljigsaw,pos=1.6}{$(\cxtnode,x)$}
        \vout{1}{115}{myblue}{smalljigsaw,pos=1.6}{$(x,\cxtnode)$}
        \draw [->, interconnect,myorange] (1) to [loop=180,looseness=22,pos=0.55] node [jigsaw,smalljigsaw] {$(x,x)$} (1);
        \draw [->, interconnect,myorange] (1) to [bend left=70] node [jigsaw,smalljigsaw] {$(x,y)$} (2);
        \draw [->, interconnect,myorange] (2) to [bend left=70] node [jigsaw,smalljigsaw] {$(y,x)$} (1);
        \vin{2}{-65}{mygreen}{smalljigsaw,pos=1.6}{$(\cxtnode,y)$}
        \vout{2}{65}{mypurple}{smalljigsaw,pos=1.6}{$(y,\cxtnode)$}
        \draw [->, interconnect,myorange] (2) to [loop=0,looseness=22,pos=0.55] node [jigsaw,smalljigsaw] {$(y,y)$} (2);
    \end{scope}
    \tikzstep{40}
    \begin{scope}[xshift=66mm]
        \node (1) [n] {};
        \vin{1}{-135}{myred}{smalljigsaw,pos=1.8,xshift=-3mm}{$(\cxtnode,x)$}
        \vout{1}{135}{myblue}{smalljigsaw,pos=1.8,xshift=-3mm}{$(x,\cxtnode)$}
        \vin{1}{-45}{mygreen}{smalljigsaw,pos=1.8,xshift=3mm}{$(\cxtnode,y)$}
        \vout{1}{45}{mypurple}{smalljigsaw,pos=1.8,xshift=3mm}{$(y,\cxtnode)$}
        \draw [->, interconnect,myorange] (1) to [thinloop=180,looseness=35,pos=0.6] node [jigsaw,smalljigsaw] {$(x,x)$} (1);
        \draw [->, interconnect,myorange] (1) to [thinloop=0,looseness=35,pos=0.6] node [jigsaw,smalljigsaw] {$(y,y)$} (1);
        \draw [->, interconnect,myorange] (1) to [thinloop=90,looseness=30,pos=0.6] node [jigsaw,smalljigsaw] {$(x,y)$} (1);
        \draw [->, interconnect,myorange] (1) to [thinloop=-90,looseness=30,pos=0.6] node [jigsaw,smalljigsaw] {$(y,x)$} (1);
    \end{scope}
  \end{tikzpicture}
  \end{align*}
  
For a more elaborate shorthand notation, see Appendix~\ref{sec:extended:shorthand:notation}.
\end{notation}

\section{Modeling Wait-For Graphs and Deadlock Detection}\label{sec:wait:for:graphs}

We now give a more extensive and more realistic modeling example that showcases the expressive power of PGR.

A \emph{wait-for graph}~\cite{fokkink2018distributed} is a hypergraph in which nodes represent processes, and hyperedges represent requests for resources. A hyperedge has a single source $p$, representing the process requesting the resources, and $M > 0$ targets distinct from $p$, representing the processes from which a resource is requested. The process $p$ requires $0 < N \leq M$ of these resources. Thus, for a fixed $M$, there are multiple types of hyperedges, each representing a particular $N$-out-of-$M$ request. Processes can 
have at most one outgoing $N$-out-of-$M$ request.

The following distributed system behavior is associated with wait-for graphs. A process without an outgoing request is said to be \emph{unblocked}. An unblocked process can \emph{grant} an incoming request, deleting the edge, or create a new $N$-out-of-$M$ request. A process becomes unblocked when its pending $N$-out-of-$M$ request is \emph{resolved}, i.e., when $N$ targeted processes have granted the request.

In order to better illustrate some of PGR's transformational power, we introduce one additional, noncanonical behavior. We consider a process $p$ \emph{overloaded} when it has $n > 2$ incoming requests. When $p$ is overloaded, a clone of $p$, $c(p)$, may be created which takes over $n - 2$ of $p$'s incoming requests. Because we assume that any outgoing request must be resolved before any incoming requests can be resolved, $c(p)$ replicates $p$'s outgoing request, if $p$ has one.

We first define a graph grammar that defines the class of valid wait-for graphs. Then, we will show how to augment the rule set in order to model the distributed system behavior. Finally, we explain how deadlocks can be detected. Throughout, we encode hypergraphs as multigraphs. Note that in this encoding, vertices representing processes are always free of loops, while vertices representing hyperedges always have loops. Hence, the given rules can appropriately discriminate between the two types of vertices.

\subsection{Wait-For Graph Grammar} The starting graph of the grammar is the empty graph, denoted by $\emptyset$. Rule
\begin{align}\tag{\textsc{create}}
  \begin{tikzpicture}[default,node distance=15mm,n/.style={graphNode},baseline=0ex]
  \begin{scope}[xshift=0mm]
    \node (1)  {$\emptyset$};
  \end{scope}
  \tikzstep{6}
  \begin{scope}[xshift=16mm]
    \node (1) [n] {};
  \end{scope}
  \end{tikzpicture}\label{wait:for:create}
\end{align}
models process creation, and rule
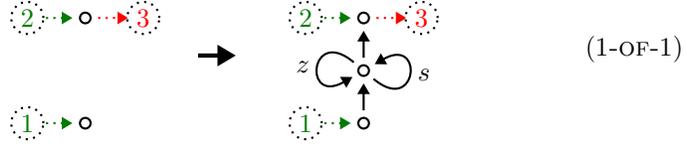
\begin{align}\tag{\textsc{$1$-of-$1$}}
  \begin{tikzpicture}[default,node distance=10mm,n/.style={graphNode},baseline=0ex]
  \begin{scope}[xshift=0mm]
    \node (1) [n,yshift=5mm] {};
    \node (2) [n,yshift=-9mm] {};
    \vin{1}{180}{mygreen}{}{$2$}
    \vout{1}{0}{myred}{}{$3$}
    \vin{2}{180}{mygreen}{}{$1$}
  \end{scope}
  \tikzstep{15}
  \begin{scope}[xshift=37mm]
    \node (1) [n,yshift=5mm] {};
    \node (2) [n,yshift=-9mm] {};
    \vin{1}{180}{mygreen}{}{$2$}
    \vout{1}{0}{myred}{}{$3$}
    \vin{2}{180}{mygreen}{}{$1$}
    \node (3) [n, yshift=-2mm] {};
    \draw [->] (2) to node {} (3); 
    \draw [->] (3) to node {} (1); 
    \draw [->, loop=-180,looseness=10] (3) to node [left] {$z$} (3);
    \draw [->, loop=0,looseness=10] (3) to node [right] {$s$} (3);
  \end{scope}
  \end{tikzpicture}\label{wait:for:one:out:of:one}
\end{align}
allows constructing a valid 1-out-of-1 request between nodes.
Labels $z$ and $s$ should be interpreted as $0$ and the successor function, respectively, so that $n$ $s$-loops encode that $n$ requests are yet to be granted.

In the grammar, any $N$-out-of-$M$ request can be extended to a valid $N$-out-of-$(M+1)$ request using rule
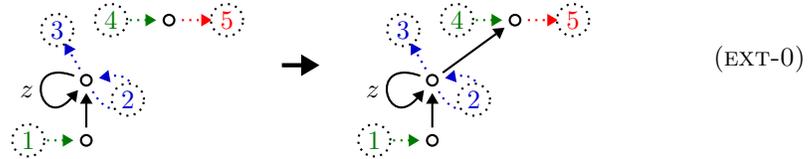
\begin{align}\tag{\textsc{ext-0}}
  \begin{tikzpicture}[default,node distance=15mm,n/.style={graphNode},baseline=0ex]
  \begin{scope}[xshift=0mm]
    \node (2) [n,yshift=-10mm] {};
    \node (3) [n,yshift=6mm,xshift=11mm] {};
    \vin{3}{180}{mygreen}{}{$4$}
    \vout{3}{0}{myred}{}{$5$}
    \vin{2}{180}{mygreen}{}{$1$}
    \node (4) [n, yshift=-2mm] {};
    \vout{4}{120}{myblue}{}{$3$}
    \draw [->] (2) to node {} (4); 
    \draw [->, loop=200,looseness=10] (4) to node [left] {$z$} (4);
    \draw [->, interconnect, loop=-20,looseness=10, myblue] (4) to node [jigsaw] {$2$} (4);
  \end{scope}
  \tikzstep{26}
  \begin{scope}[xshift=46mm]
    \node (2) [n,yshift=-10mm] {};
    \node (3) [n,yshift=6mm,xshift=11mm] {};
    \vin{3}{180}{mygreen}{}{$4$}
    \vout{3}{0}{myred}{}{$5$}
    \vin{2}{180}{mygreen}{}{$1$}
    \node (4) [n, yshift=-2mm] {};
    \vout{4}{120}{myblue}{}{$3$}
    \draw [->] (2) to node {} (4); 
    \draw [->, loop=200,looseness=10] (4) to node [left] {$z$} (4);
    \draw [->, interconnect, loop=-20,looseness=10, myblue] (4) to node [jigsaw] {$2$} (4);
    \draw [->] (4) to node {} (3);
  \end{scope}
  \end{tikzpicture}\label{wait:for:ext:zero}
\end{align}
and to a valid $(N+1)$-out-of-$(M+1)$ request using rule
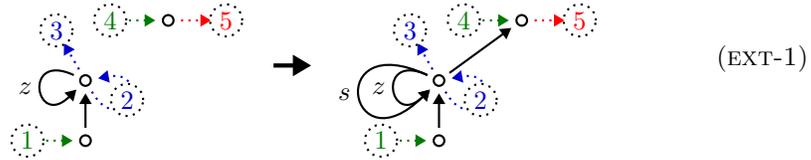
\begin{align}\tag{\textsc{ext-1}}
  \begin{tikzpicture}[default,node distance=15mm,n/.style={graphNode},baseline=0ex]  \begin{scope}[xshift=0mm]
    \node (2) [n,yshift=-10mm] {};
    \node (3) [n,yshift=6mm,xshift=11mm] {};
    \vin{3}{180}{mygreen}{}{$4$}
    \vout{3}{0}{myred}{}{$5$}
    \vin{2}{180}{mygreen}{}{$1$}
    \node (4) [n, yshift=-2mm] {};
    \vout{4}{120}{myblue}{}{$3$}
    \draw [->] (2) to node {} (4); 
    \draw [->, loop=195,looseness=10] (4) to node [left] {$z$} (4);
    \draw [->, interconnect, loop=-20,looseness=10, myblue] (4) to node [jigsaw] {$2$} (4);
  \end{scope}
  \tikzstep{25}
  \begin{scope}[xshift=47mm]
    \node (2) [n,yshift=-10mm] {};
    \node (3) [n,yshift=6mm,xshift=11mm] {};
    \vin{3}{180}{mygreen}{}{$4$}
    \vout{3}{0}{myred}{}{$5$}
    \vin{2}{180}{mygreen}{}{$1$}
    \node (4) [n, yshift=-2mm] {};
    \vout{4}{120}{myblue}{}{$3$}
    \draw [->] (2) to node {} (4); 
    \draw [->, loop=195,looseness=10] (4) to node [left] {$z$} (4);
    \draw [->, interconnect, loop=-20,looseness=10, myblue] (4) to node [jigsaw] {$2$} (4);
    \draw [->] (4) to node {} (3);
    \draw [->, loop=195,looseness=20] (4) to node [left] {$s$} (4);
  \end{scope}
  \end{tikzpicture}\label{wait:for:ext:one}
\end{align}
These four rules suffice for generating any valid wait-for graph.

\subsection{System Modeling}
To model a distributed system, we need rule~\ref{wait:for:create} for process creation, as well as its inverse, \textsc{destroy}, for process destruction. Note that \textsc{destroy} 
constrains the process selected for destruction to be isolated in our framework (i.e., it is not associated with any pending requests), as desired.

Any $N$-out-of-$M$ request is understood to be an atomic action. So for, e.g., modeling 2-out-of-2 requests, we need the rule
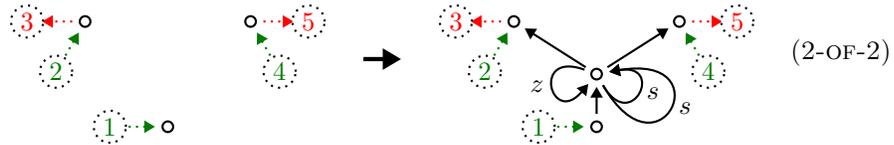
\begin{align}\tag{\textsc{2-of-2}}
  \begin{tikzpicture}[default,node distance=15mm,n/.style={graphNode},baseline=0ex]
  \begin{scope}[xshift=0mm]
    \node (1) [n,yshift=5mm,xshift=-11mm] {};
    \node (2) [n,yshift=-9mm] {};
    \node (3) [n,yshift=5mm,xshift=11mm] {};
    \vin{1}{-120}{mygreen}{}{$2$}
    \vout{1}{180}{myred}{}{$3$}
    \vin{3}{-60}{mygreen}{}{$4$}
    \vout{3}{0}{myred}{}{$5$}
    \vin{2}{180}{mygreen}{}{$1$}
  \end{scope}
  \tikzstep{26}
    \begin{scope}[xshift=57mm]
    \node (1) [n,yshift=5mm,xshift=-11mm] {};
    \node (2) [n,yshift=-9mm] {};
    \node (3) [n,yshift=5mm,xshift=11mm] {};
    \vin{1}{-120}{mygreen}{}{$2$}
    \vout{1}{180}{myred}{}{$3$}
    \vin{3}{-60}{mygreen}{}{$4$}
    \vout{3}{0}{myred}{}{$5$}
    \vin{2}{180}{mygreen}{}{$1$}
    \node (4) [n, yshift=-2mm] {};
    \draw [->] (2) to node {} (4); 
    \draw [->] (4) to node {} (1);
    \draw [->] (4) to node {} (3);
    \draw [->, loop=200,looseness=10] (4) to node [left] {$z$} (4);
    \draw [->, loop=-20,looseness=10] (4) to node [right] {$s$} (4);
    \draw [->, loop=-20,looseness=20] (4) to node [right] {$s$} (4);
  \end{scope}
  \end{tikzpicture}\label{wait:for:two:out:of:two}
\end{align}
\vspace{-4ex}

\noindent
Such rules can be easily simulated by a contiguous sequence of rewrite steps $\ref{wait:for:one:out:of:one} \cdot \ref{wait:for:ext:zero}^* \cdot \ref{wait:for:ext:one}^*$, in which the node making the request remains fixed. We omit the details.

A grant request may be modeled by
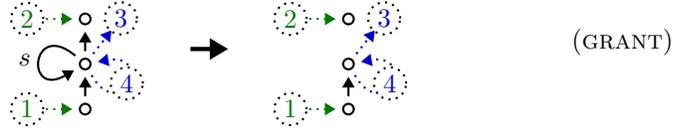
\begin{align}\tag{\textsc{grant}}
  \begin{tikzpicture}[default,node distance=15mm,n/.style={graphNode},baseline=0ex]
  \begin{scope}[xshift=0mm]
    \node (1) [n,yshift=4mm] {};
    \node (2) [n,yshift=-8mm] {};
    \vin{1}{180}{mygreen}{}{$2$}
    \vin{2}{180}{mygreen}{}{$1$}
    \node (3) [n, yshift=-2mm] {};
    \draw [->] (2) to node {} (3); 
    \draw [->] (3) to node {} (1); 
    \draw [->, loop=-180,looseness=10] (3) to node [left] {$s$} (3);
    \vout{3}{52}{myblue}{}{$3$}
    \draw [->, interconnect, loop=-20,looseness=10, myblue] (3) to node [jigsaw] {$4$} (3);
  \end{scope}
  \tikzstep{14}
  \begin{scope}[xshift=35mm]
    \node (1) [n,yshift=4mm] {};
    \node (2) [n,yshift=-8mm] {};
    \vin{1}{180}{mygreen}{}{$2$}
    \vin{2}{180}{mygreen}{}{$1$}
    \node (3) [n, yshift=-2mm] {};
    \draw [->] (2) to node {} (3);
    \vout{3}{52}{myblue}{}{$3$}
    \draw [->, interconnect, loop=-20,looseness=10, myblue] (3) to node [jigsaw] {$4$} (3);
  \end{scope}
  \end{tikzpicture}\label{wait:for:grant}
\end{align}
and a request resolution by
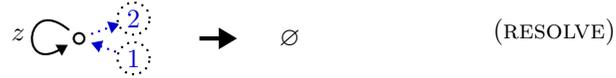
\begin{align}\tag{\textsc{resolve}}
  \begin{tikzpicture}[default,node distance=15mm,n/.style={graphNode},baseline=0ex]
  \begin{scope}[xshift=0mm]
    \node (1) [n] {};
    \draw [->, loop=180,looseness=10] (1) to node [left] {$z$} (1);
    \vin{1}{-20}{myblue}{}{$1$}
    \vout{1}{20}{myblue}{}{$2$}
  \end{scope}
  \tikzstep{16}
  \begin{scope}[xshift=28mm]
    \node (1) [] {$\emptyset$};
  \end{scope}
  \end{tikzpicture}\label{wait:for:resolve}
\end{align}

This leaves only the cloning behavior for an overloaded process $p$. This requires two rules: one for the case in which $p$ is unblocked, and one for the case in which $p$ is blocked. We use the shorthand notation introduced in Notation~\ref{notation:shorthand}, so that named nodes $r_i$ induce type edges among themselves and from and into the context.\footnote{The type edges between distinctly named nodes $r_i \neq r_j$ are redundant in the considered scenario, since we know that these type edges will never have adherents.}

The case in which $p$ is unblocked is modeled by rule
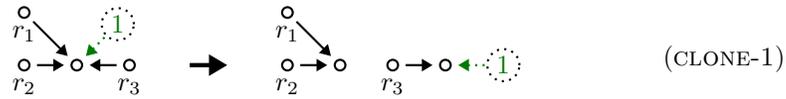
\begin{align}\tag{\textsc{clone-1}}
  \begin{tikzpicture}[default,node distance=7mm,n/.style={graphNode},baseline=0ex]
  \begin{scope}[xshift=0mm]
    \node (p) [n] {};
    \node[left of=p] (r2) [n] {};
        \node [anchor=center,at=(r2.south),inner sep=0,yshift=-1.1mm] {$r_2$};
    \node[above of=r2](r1) [n] {};
        \node [anchor=center,at=(r1.south),inner sep=0,yshift=-1.1mm] {$r_1$};
    \node[right of=p] (r3) [n] {};
        \node [anchor=center,at=(r3.south),inner sep=0,yshift=-1.1mm] {$r_3$};
    \draw[->] (r1) to node {} (p);
    \draw[->] (r2) to node {} (p);
    \draw[->] (r3) to node {} (p);
    \vin{p}{45}{mygreen}{}{$1$}
  \end{scope}
  \tikzstep{15}
  \begin{scope}[xshift=35mm]
    \node (p) [n] {};
    \node[left of=p] (r2) [n] {};
        \node [anchor=center,at=(r2.south),inner sep=0,yshift=-1.1mm] {$r_2$};
    \node[above of=r2] (r1) [n] {};
        \node [anchor=center,at=(r1.south),inner sep=0,yshift=-1.1mm] {$r_1$};
    \node[right of=p] (r3) [n] {};
        \node [anchor=center,at=(r3.south),inner sep=0,yshift=-1.1mm] {$r_3$};
    \node[right of=r3] (cp) [n] {}; %
    \draw[->] (r1) to node {} (p);
    \draw[->] (r2) to node {} (p);
    \draw[->] (r3) to node {} (cp);
    \vin{cp}{0}{mygreen}{}{$1$} %
  \end{scope}
  \end{tikzpicture}\label{wait:for:clone1}
 \end{align}
and the case in which $p$ is blocked is modeled by rule
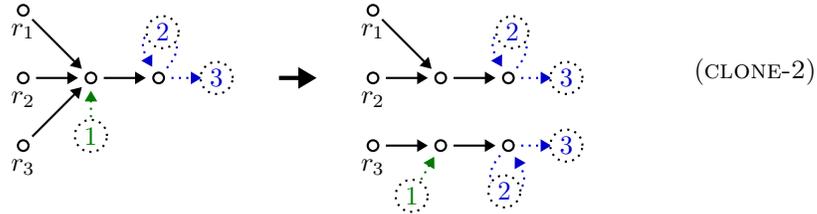
\begin{align}\tag{\textsc{clone-2}}
  \begin{tikzpicture}[default,node distance=9mm,n/.style={graphNode},baseline=0ex]
  \begin{scope}[xshift=0mm]
    \node (p) [n] {};
    \node[left of=p] (r2) [n] {};
        \node [anchor=center,at=(r2.south),inner sep=0,yshift=-1.1mm] {$r_2$};
    \node[above of=r2](r1) [n] {};
        \node [anchor=center,at=(r1.south),inner sep=0,yshift=-1.1mm] {$r_1$};
    \node[below of=r2] (r3) [n] {};
        \node [anchor=center,at=(r3.south),inner sep=0,yshift=-1.1mm] {$r_3$};
    \draw[->] (r1) to node {} (p);
    \draw[->] (r2) to node {} (p);
    \draw[->] (r3) to node {} (p);
    \vin{p}{-90}{mygreen}{}{$1$}
    \node[right of=p] (r) [n] {};
    \draw[->] (p) to node {} (r);
    \draw [->, interconnect, loop=90,looseness=10, myblue] (r) to node [jigsaw] {$2$} (r);
    \vout{r}{0}{myblue}{}{$3$}
  \end{scope}
  \tikzstep{25}
  \begin{scope}[xshift=46.5mm]
    \node (p) [n] {};
    \node[left of=p] (r2) [n] {};
        \node [anchor=center,at=(r2.south),inner sep=0,yshift=-1.1mm] {$r_2$};
    \node[above of=r2] (r1) [n] {};
        \node [anchor=center,at=(r1.south),inner sep=0,yshift=-1.1mm] {$r_1$};
    \node[below of=r2] (r3) [n] {};
        \node [anchor=center,at=(r3.south),inner sep=0,yshift=-1.1mm] {$r_3$};
    \node[below of=p] (cp) [n] {}; %
    \draw[->] (r1) to node {} (p);
    \draw[->] (r2) to node {} (p);
    \draw[->] (r3) to node {} (cp);
    \vin{cp}{-120}{mygreen}{}{$1$} %
    \node[right of=p] (r) [n] {};
    \draw[->] (p) to node {} (r);
    \draw [->, interconnect, loop=90,looseness=10, myblue] (r) to node [jigsaw] {$2$} (r);
    \vout{r}{0}{myblue}{}{$3$}
    \node[right of=cp] (cpr) [n] {};%
    \draw[->] (cp) to node {} (cpr);
    \draw [->, interconnect, loop=-90,looseness=10, myblue] (cpr) to node [jigsaw] {$2$} (cpr);
    \vout{cpr}{0}{myblue}{}{$3$}
  \end{scope}
  \end{tikzpicture}\label{wait:for:clone2}
 \end{align}

Cloning would be easier to express if PGR were to be extended with support for hyperedges and cardinality constrained type edges. We envision a rule like
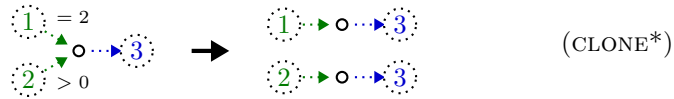
\begin{align}\tag{\textsc{clone*}}
  \begin{tikzpicture}[default,node distance=7mm,n/.style={graphNode},baseline=0ex]
  \begin{scope}[xshift=0mm]
    \node (p) [n] {};
    \vin{p}{150}{mygreen}{}{$1$}
    \vin{p}{-150}{mygreen}{}{$2$}
    \vout{p}{0}{myblue}{}{$3$}
    \node[above of=p, yshift=-2.5mm, xshift=-0.8mm] (c1) [] {{\scriptsize $ =2$}};
    \node[below of=p, yshift=2.7mm, xshift=-0.8mm] (c1) [] {{\scriptsize $>0$}};
  \end{scope}
  \tikzstep{15}
  \begin{scope}[xshift=35mm]
    \node[yshift=3.5mm] (p) [n] {};
    \node[below of=p] (cp) [n] {};
    \vin{p}{180}{mygreen}{}{$1$}
    \vin{cp}{-180}{mygreen}{}{$2$}
    \vout{p}{0}{myblue}{}{$3$}
    \vout{cp}{0}{myblue}{}{$3$}
  \end{scope}
  \end{tikzpicture}\label{wait:for:clone:extended}
 \end{align}
 to capture the same cases as rules \ref{wait:for:clone1} and \ref{wait:for:clone2} combined.
 We leave such an extension to future work. In particular, the precise semantics of hyperedge transformation would have to be determined.

\subsection{Deadlock Detection} Deadlock detection on some valid wait-for graph $G$ can be performed by restricting the rewrite system to rules \ref{wait:for:grant}, \ref{wait:for:resolve} and \textsc{destroy}, yielding a terminating rewrite system. Then the network represented by $G$ contains a deadlock if and only if the (unique) normal form of $G$ is the empty graph $\emptyset$.

For a comparable modeling example, see Appendix~\ref{sec:terminiation:dijkstra:scholten}, in which the Dijkstra--Scholten termination detection algorithm is modeled.

\section{Comparison}\label{sec:comparison}

\newcommand{\framework}[1]{\medskip\noindent\textbf{#1.}\; }
We compare PGR to several other rewriting frameworks. We have selected these frameworks because of their popularity and/or because they bear certain similarities to our approach.

\framework{Double-Pushout (DPO)}
Ehrig et al.'s double-pushout approach (DPO)~\cite{ehrig1973graph} is one of the most prominent approaches in graph rewriting.

A rewrite rule in the DPO approach is of the form $L \hookleftarrow K \rightarrow R$, where $L$ is the 
subgraph to be replaced by subgraph $R$. The graph $K$ is an ``interface'', used to identify a part of $L$ with a part of $R$, and it can be thought of as describing which part of $L$ is preserved by the rule. The identification is formally established through the inclusion $L \hookleftarrow K$ and the graph morphism $K \to R$. The morphism $K \to R$ may be 
non-injective, allowing it to merge nodes that are in the interface.

A DPO rewrite rule $L \stackrel{\varphi}{\hookleftarrow} K \stackrel{\psi}{\rightarrow} R$ is applied inside a graph $G$ as follows~\cite{ehrig86tutorial,ehrig90tutorial}. Let $m : L \to G$ be a graph morphism, which we may assume to be injective~\cite{habel2001double}. The graph $m(L) \isomorph L$ is said to be a \emph{match} for $L$.
The arising rewrite step 
replaces $m(L)$ of $G$
by a fresh copy $c(R)$ of $R$, redirecting edges left dangling by the removal of a $v \in m(L)$ to node $c(\psi(\varphi^{-1}(m^{-1}(v))))$.
For the redirection of edges to work, nodes that leave dangling ends need to be part of the interface, that is, in $m(\varphi(K))$.
This is known as the ``gluing condition''.\footnote{%
By the injectivity assumption for $m$, we need not consider what is known as the ``identification condition''.}
If the gluing condition is not met, the rewrite step is not permitted.

Using Notation~\ref{notation:shorthand}, it is easy to see that PGR at least as expressive as DPO with respect to the generated rewrite relation.
A DPO rule $L \stackrel{\varphi}{\hookleftarrow} K \stackrel{\psi}{\rightarrow} R$ can be directly simulated by a PGR rule $L \to R$ in which the nodes are annotated with their (set of) names in the interface:
$v \in V_L$ is annotated with the names $\varphi^{-1}(v)$, and
$v \in V_R$ is annotated with the names $\psi^{-1}(v)$.

However, DPO is stronger in one respect: a DPO rewrite step preserves the subgraph specified in $K$, whereas a PGR rewrite can be thought of as destroying and replacing the left-hand side of the rule. The consequences for metaproperties relating to parallelism and concurrency will have to be investigated. 

\framework{Generalized DPO}
In some variants of DPO, the inclusion $L \hookleftarrow K$ is generalized to a (possibly non-injective) morphism $\varphi: K \to L$. Intuitively, this allows a node $v$ of $L$ to be ``split apart'' in the interface $K$. Applying the DPO method to a host graph ensures that the patch graph edges incident to $v$ will be incident to one of $v$'s split copies. It does not dictate how these edges should be distributed. Thus, such rewrite steps are non-deterministic.

Generalized DPO rules $L \stackrel{\varphi}{\leftarrow} K \stackrel{\psi}{\rightarrow} R$ can be translated to PGR rules $L \to R$ in the same way as discussed for DPO. Since $\varphi$ is no longer required to be injective, nodes $v \in V_L$ can be annotated with multiple names $\varphi^{-1}(v)$, thereby leading to (non-deterministic) quasi rules (Definition~\ref{def:graph:rewrite:rule}).See also Example~\ref{ex:splitting:node}.

\framework{Single-Pushout (SPO)}
The single-pushout (SPO) approach by L\"owe~\cite{lowe1993algebraic} is the destructive sibling of DPO. It is operationally like DPO, but it drops the gluing condition. Any edges that would become dangling in the host graph are instead removed.

An SPO rule $L \stackrel{\varphi}{\hookleftarrow} K \stackrel{\psi}{\rightarrow} R$ can be simulated by a PGR rule $L \to R$ with annotations as discussed above for DPO, except that each node
$v \in V_L$, for which $\varphi^{-1}(v)$ is empty, is now annotated by a fresh name. The rewrite step will then delete all patch edges connected to such a node.

\framework{DPO Rewriting in Context (DPO-C)}
The DPO Rewriting in Context (DPO-C) approach by L\"owe~\cite{lowe2018double,lowe2019double} addresses the issue of non-determinism in generalized variants of DPO, using ingoing and outgoing arrow annotations to dictate how these edges should be distributed over split copies. The visual representation of DPO-C therefore bears some similarity to that of PGR. In addition, absence of arrow annotations also define negative application conditions like in PGR. However, the patch cannot be transformed as freely as in PGR. For instance, see rule~\eqref{rule:agree} below. 

\framework{AGREE} 
AGREE~\cite{corradini2015agree} and PBPO \cite{corradini2019pbpo} by Corradini et al.\ extend DPO with the ability to erase and clone nodes, while also being able to specify how patch edges are distributed among the copies. For this purpose, a ``filter'' for the edges determines what kind of patch edges are to be dropped. This filter can distinguish different types of edges based on their source, target and label.
Thereby AGREE and PBPO subsume mildly restricted versions of DPO, SPO, and other formalisms.

PGR has some features that are not present in  AGREE and PBPO.
First, in PGR rule applicability can be restricted by conditions on the permitted shape of the patch. 
Second, PGR allows (almost) arbitrary redirecting, moving and copying of patch edges outside the scope of AGREE and PBPO.
For instance,
  \begin{align}
  \begin{tikzpicture}[default,node distance=24mm,n/.style={graphNode},baseline=-.6ex]
    \begin{scope}[xshift=0mm]
        \node (1) [n] {};
        \node (2) [n,right of=1] {};
        \vin{1}{-180}{myred}{}{$1$}
        \vout{1}{0}{myblue}{}{$2$}
        \vin{2}{-180}{mygreen}{}{$6$}
        \vout{2}{0}{mypurple}{}{$7$}
    \end{scope}
     \tikzstep{39}
    \begin{scope}[xshift=58mm]
        \node (1) [n] {};
        \node (2) [n,right of=1] {};
        \vin{2}{-180}{myred}{}{$1$}
        \vout{1}{0}{myblue}{}{$2$}
        \vin{1}{-180}{mygreen}{}{$6$}
        \vout{2}{0}{mypurple}{}{$7$}
    \end{scope}
  \end{tikzpicture}
  \label{rule:agree}
  \end{align}
cannot be expressed in the latter frameworks.
Also inverting the direction of patch edges,
or moving patch edges between nodes of the pattern is not possible in AGREE and PBPO.

On the other hand, AGREE and PBPO capture some transformations that cannot be expressed in PGR.
First, AGREE and PBPO can express some global transformations, unlike PGR.
Second, the ``patch edge filter'' in AGREE and PBPO can distinguish patch edges depending on their label and the ``type'' of the source/target in the context (here the type is given by some type graph).
Both features are not present in PGR as presented in this paper.
However, PGR can be extended with constraints on the patch type edges. (For a discussion on how to encode constraints in the present framework, see Appendix~\ref{sec:richer:constraints}.) We leave the investigation of a suitable constraint language to future work.

\framework{Nagl's Approach} 
Nagl~\cite{nagl1986set} has defined a very powerful graph transformation approach. Rather than identifying ``gluing points'' for the replacement of $L$ by $R$ in a host graph $G$ (as in the previous approaches), rules are equipped with expressions that describe how $R$ is to be embedded into the remainder graph $G^{-} = G - L$. An expression can, e.g., specify that an edge must be created from $u \in G^{-}$ to $v \in R$ if there existed a path (of certain length and containing certain labels) from $u$ to a $w \in L$. Thus, the embedding context may no longer even be local.

While not all of these transformations are supported by PGR, the expressions are arguably much less intuitive than our representation, in which both application conditions and transformations are visualized in a unified manner.

\framework{Habel et al.'s Approach} 
Habel et al.~\cite{habel96graph} have introduced  graph grammars with rule-dependent application conditions that also admit a very intuitive visual representation. These conditions are more powerful than PGR's application conditions, since they can extend arbitrarily far into the context graph. However, transformations are not included in the approach, unlike in PGR, in which the notations for application conditions and transformations are unified.

\framework{Drags} 
To generalize term rewriting techniques to the domain of graphs, Dershowitz and Jouannaud~\cite{dershowitz2019drags} have recently defined the \emph{drag} data structure and framework. A drag is a multigraph in which nodes are labeled by symbols that have an arity equal to the node's outdegree. Nodes labeled with nullary variable symbols are called \emph{sprouts}, and resemble output ports. In addition, the drag comes equipped with a finite number of \emph{roots}, which resemble input ports.

A composition operation $\otimes_\xi$ for drags, parameterized by a two-way \emph{switchboard}~$\xi$ identifying sprouts with roots, gives rise to a rewrite relation $W \otimes_\xi L \to W \otimes_\xi R$. For this rewrite relation to be well-defined, it is required, among others, that $L$ and $R$ have the same number of roots and the same multisets of variables.

Since drag rewriting imposes arity restrictions on nodes, it is more restrictive than patch rewriting concerning the shapes of the graphs that can be rewritten. As drag rewrite steps are local, we believe that PGR can simulate them, but we leave this investigation to future work.

\section{Conclusion}\label{sec:conclusion}

We have introduced \emph{patch graph rewriting}, a framework for graph rewriting that enriches the rewrite rules with a simple, yet powerful language for constraining and transforming the local embedding---or \emph{patch}.  

For future work, we plan to investigate various meta-properties central in graph rewriting~\cite{ehrig90tutorial}, in particular confluence~\cite{plump2005,heckel2002,overbeek2020}, termination~\cite{dershowitz2018graph,zantema2014}, the concurrency theorem~\cite{ehrig1980parallelism}, decomposability and reversibility of rules. We intend to study these properties both globally, for all graphs, as well as locally~\cite{termination:local:2010,termination:automata:2015}, for a given language of graphs~\cite{rensink2004canonical}.
Furthermore, we are interested in extending the framework with constraint labels on patch type edges, and in allowing label transformations. We believe this could be useful for modeling a larger class of distributed algorithms~\cite{fokkink2018distributed}.
Another interesting direction of research is an equational perspective on patch rewriting, as similarly investigated by Ariola and Klop for term graph rewriting~\cite{ariola1996equational}.

\subsubsection*{Acknowledgments} This paper has benefited from discussions with Jan Willem Klop, Nachum Dershowitz, Femke van Raamsdonk, Roel de Vrijer, and Wan Fokkink. We thank Andrea Corradini and the anonymous reviewers for their useful suggestions. Both authors received funding from the Netherlands Organization for Scientific Research (NWO) under the Innovational Research Incentives Scheme Vidi (project.\ No.\ VI.Vidi.192.004).

\bibliographystyle{plainurl}
\bibliography{main}

\newpage 
\section*{Appendix}
\appendix

\section{Depicting Graphs}\label{sec:depicting:graphs}

  Consider the graph $G = (V,E,\ssrc,\stgt,\slab)$ defined by:
  \begin{align*}
      V &= \set{ 1,2,3,4 } 
        & \src{e_1} &= 1 & \tgt{e_1} &= 2 &  & \lab{e_1} = a \\
      E &= \set{ e_1,e_2,e_3,e_4 }
        & \src{e_2} &= 2 & \tgt{e_2} &= 3 &  & \lab{e_2} = b \\
      L &= \set{ a,b }
        & \src{e_3} &= 2 & \tgt{e_3} &= 3 &  & \lab{e_3} = b \\
      & & \src{e_4} &= 4 & \tgt{e_4} &= 2 &  & \lab{e_4} = a
  \end{align*}
  The graph $G$ can be visualized as follows:
  
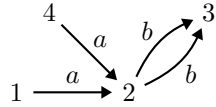
\begin{center}
  \begin{tikzpicture}[default,node distance=15mm]
    \node (1) {$1$};
    \node (2) [right of=1] {$2$};
    \node (3) [above right of=2] {$3$};
    \node (4) [above left of=2] {$4$};
    
    \draw [->] (1) to node [above] {$a$} (2);
    \draw [->] (2) to[bend right=20] node [below right] {$b$} (3);
    \draw [->] (2) to[bend right=-20] node [above left] {$b$} (3);
    \draw [->] (4) to node [above right] {$a$} (2);
  \end{tikzpicture}
  \end{center}
  The edges are displayed as arrows from the source to the target of the edge,
  annotated by the label of the edge.
  In such a visualization names of the edges are typically suppressed;
  so the graph is defined only up to renaming of the edges.
  
  We also suppress the vertex names if they are irrelevant.
  Then the graphs are defined up to isomorphism only.
  For instance,
  
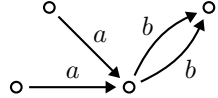
\begin{center}
  \begin{tikzpicture}[default,node distance=15mm,n/.style={graphNode}]
    \node (1) [n] {};
    \node (2) [right of=1,n] {};
    \node (3) [above right of=2,n] {};
    \node (4) [above left of=2,n] {};
    
    \draw [->] (1) to node [above] {$a$} (2);
    \draw [->] (2) to[bend right=20] node [below right] {$b$} (3);
    \draw [->] (2) to[bend right=-20] node [above left] {$b$} (3);
    \draw [->] (4) to node [above right] {$a$} (2);
  \end{tikzpicture}
  \end{center}
  is a visualisation of $G$ where edge and vertex names are suppressed. As a convention, nodes with distinct coordinates are always assumed to be distinct.

\section{Encoding Vertex Labels}\label{sec:vertex:labels}

Vertex labels can be encoded in  at least two different ways:
  \begin{enumerate}[label=(\roman*)]
    \item
      Choosing a vertex label set $L_V$ disjoint from edge label set $L_E$, and adding an edge $v \xrightarrow{\alpha} v$ when $v$ has label $\alpha$. 
    \item
      Using a distinguished node $r$ that
      has precisely one edge $r \xrightarrow{\alpha} v$
      to every other node $v$ and $\alpha$ is $v$'s vertex label. Then $r$ is the only node with no incoming edges. 
  \end{enumerate}
  In PGR, approach (ii) has an advantage over (i) 
  when one wants to match all loops of a node with an arbitrary label.
  Assume that we want to specify a rule that can drop all loops from an arbitrary node.
  Using approach~(i) we need a rule of the form
  
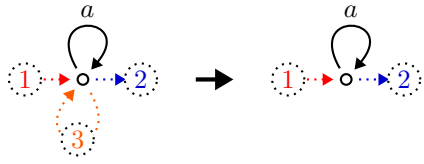
\begin{center}
  \begin{tikzpicture}[default,node distance=20mm,n/.style={graphNode}]
    \begin{scope}[xshift=0mm]
        \node (1) [n] {};
        \draw [<-] (1) to [loop=90,looseness=12] node [above] {$a$} (1);
        \vin{1}{180}{myred}{}{$1$}
        \vout{1}{0}{myblue}{}{$2$}
        \draw [<-, interconnect,myorange] (1) to [loop=-90,looseness=14] node [jigsaw] {$3$} (1);
    \end{scope}
    \tikzstep{15}
    \begin{scope}[xshift=35mm]
        \node (1) [n] {};
        \draw [<-] (1) to [loop=90,looseness=12] node [above] {$a$} (1);
        \vin{1}{180}{myred}{}{$1$}
        \vout{1}{0}{myblue}{}{$2$}
    \end{scope}
  \end{tikzpicture}
  \end{center}
  for every node label $a$.
  Using approach~(ii) a single rule suffices:
  
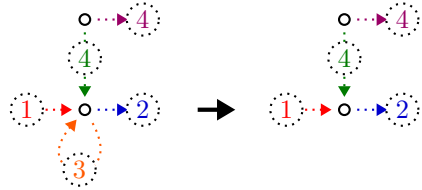
\begin{center}
  \begin{tikzpicture}[default,node distance=12mm,n/.style={graphNode}]
    \begin{scope}[xshift=0mm]
        \node (1) [n] {};
        \node (2) [n,above of=1] {};
        \draw [->,interconnect,mygreen] (2) to node [jigsaw,pos=0.4] {$4$} (1);
        \vin{1}{180}{myred}{}{$1$}
        \vout{1}{0}{myblue}{}{$2$}
        \draw [<-, interconnect,myorange] (1) to [loop=-90,looseness=14] node [jigsaw] {$3$} (1);
        \vout{2}{0}{mypurple}{}{$4$}
    \end{scope}
    \tikzstep{15}
    \begin{scope}[xshift=34.5mm]
        \node (1) [n] {};
        \node (2) [n,above of=1] {};
        \draw [->,interconnect,mygreen] (2) to node [jigsaw,pos=0.4] {$4$} (1);
        \vin{1}{180}{myred}{}{$1$}
        \vout{1}{0}{myblue}{}{$2$}
        \vout{2}{0}{mypurple}{}{$4$}
    \end{scope}
  \end{tikzpicture}
  \end{center}
  Note that the upper node cannot have incoming patch edges, so this is the node responsible for assigning labels to the other nodes. Thus \jigsawlabel{4}{mygreen} binds the edge that carries the node label of the lower node.

\section{Elementary Graph Operations}\label{sec:example:rules}

The following examples demonstrate that PGR easily supports a number of elementary graph operations.

\begin{example}[Merging Two Nodes]
  The following rewrite rule can be applied to any pair of nodes
  connected by an edge with label $a$:
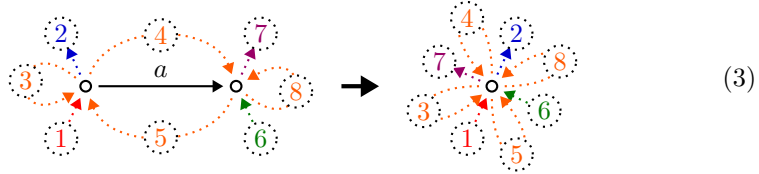
  \begin{align}
  \begin{tikzpicture}[default,node distance=20mm,n/.style={graphNode},baseline=0ex]
    \begin{scope}[xshift=0mm]
        \node (1) [n] {};
        \node (2) [n,right of=1] {};
        \draw [->] (1) to[bend right=0] node [above] {$a$} (2);
        \vin{1}{-115}{myred}{}{$1$}
        \vout{1}{115}{myblue}{}{$2$}
        \draw [->, interconnect,myorange] (1) to [loop=180,looseness=14] node [jigsaw] {$3$} (1);
        \draw [->, interconnect,myorange] (1) to [bend left=70] node [jigsaw] {$4$} (2);
        \draw [->, interconnect,myorange] (2) to [bend left=70] node [jigsaw] {$5$} (1);
        \vin{2}{-65}{mygreen}{}{$6$}
        \vout{2}{65}{mypurple}{}{$7$}
        \draw [->, interconnect,myorange] (2) to [loop=0,looseness=14] node [jigsaw] {$8$} (2);
    \end{scope}
    \tikzstep{34}
    \begin{scope}[xshift=54mm]
        \node (1) [n] {};
        \vin{1}{-115}{myred}{}{$1$}
        \vout{1}{-115+180}{myblue}{}{$2$}
        \vin{1}{-115+90}{mygreen}{}{$6$}
        \vout{1}{-115-90}{mypurple}{}{$7$}
        \draw [->, interconnect,myorange] (1) to [thinloop=205,looseness=25] node [jigsaw] {$3$} (1);
        \draw [->, interconnect,myorange] (1) to [thinloop=115,looseness=25] node [jigsaw] {$4$} (1);
        \draw [->, interconnect,myorange] (1) to [thinloop=25,looseness=25] node [jigsaw] {$8$} (1);        
        \draw [->, interconnect,myorange] (1) to [thinloop=-65,looseness=25] node [jigsaw] {$5$} (1);        
    \end{scope}
  \end{tikzpicture}
  \label{rule:merge}
  \end{align}
  When the rule is applied, the edge with label $a$ is dropped,
  and the two nodes are merged into a single node. All incoming and outgoing edges are redirected accordingly.

  If we want to exclude edges between the nodes of the pattern
  other than the edge labeled with $a$, 
  then the rule can be simplified as follows:
  
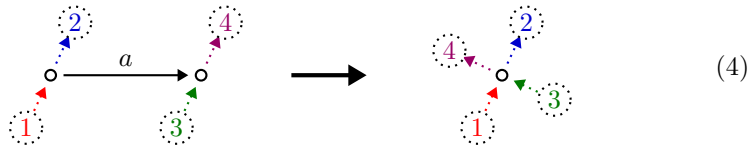
\begin{align}
  \begin{tikzpicture}[default,node distance=20mm,n/.style={graphNode},baseline=0ex]
    \begin{scope}[xshift=0mm]
        \node (1) [n] {};
        \node (2) [n,right of=1] {};
        \draw [->] (1) to node [above] {$a$} (2);
        \vin{1}{-115}{myred}{}{$1$}
        \vout{1}{-115+180}{myblue}{}{$2$}
        \vin{2}{-115}{mygreen}{}{$3$}
        \vout{2}{-115+180}{mypurple}{}{$4$}
    \end{scope}
    \draw [->,ultra thick,n/.style={graphNode}] (32mm,0) to ++(10mm,0mm);
    \begin{scope}[xshift=60mm]
        \node (1) [n] {};
        \vin{1}{-115}{myred}{}{$1$}
        \vout{1}{-115+180}{myblue}{}{$2$}
        \vin{1}{-115+90}{mygreen}{}{$3$}
        \vout{1}{-115-90}{mypurple}{}{$4$}
    \end{scope}
  \end{tikzpicture}
  \end{align}
  Now the patch can only contain edges between the context and the pattern of the rule.

\end{example}

\begin{example}[Copying a Node]\label{ex:copying:node}
  The following graph rewrite rule copies including all its edges
  (from the context, to the context and loops):
  
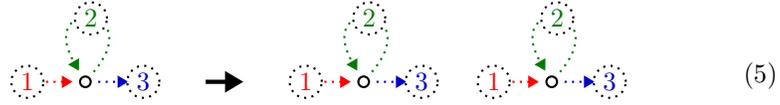
\begin{align}
  \begin{tikzpicture}[default,node distance=25mm,n/.style={graphNode},baseline=0ex]
    \begin{scope}[xshift=0mm]
        \node (1) [n] {};
        \vin{1}{180}{myred}{}{$1$}
        \draw [->, interconnect,mygreen] (1) to [loop=90,looseness=15] node [jigsaw] {$2$} (1);
        \vout{1}{0}{myblue}{}{$3$}
    \end{scope}
    \tikzstep{16}
    \begin{scope}[xshift=37mm]
        \node (1) [n] {};
        \node (2) [n,right of=1] {};
        \vin{1}{180}{myred}{}{$1$}
        \draw [->, interconnect,mygreen] (1) to [loop=90,looseness=15] node [jigsaw] {$2$} (1);
        \vout{1}{0}{myblue}{}{$3$}
        \vin{2}{180}{myred}{}{$1$}
        \draw [->, interconnect,mygreen] (2) to [loop=90,looseness=15] node [jigsaw] {$2$} (2);
        \vout{2}{0}{myblue}{}{$3$}
    \end{scope}
  \end{tikzpicture}
  \label{rule:copy}
  \end{align}
  PGR allows a fine-grained control
  that enables one to do much more than simply copying a node.
  For instance, consider the rule
  
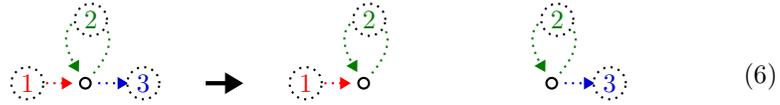
\begin{align}
  \begin{tikzpicture}[default,node distance=25mm,n/.style={graphNode},baseline=0ex]
    \begin{scope}[xshift=0mm]
        \node (1) [n] {};
        \vin{1}{180}{myred}{}{$1$}
        \draw [->, interconnect,mygreen] (1) to [loop=90,looseness=15] node [jigsaw] {$2$} (1);
        \vout{1}{0}{myblue}{}{$3$}
    \end{scope}
    \tikzstep{16}
    \begin{scope}[xshift=37mm]
        \node (1) [n] {};
        \node (2) [n,right of=1] {};
        \vin{1}{180}{myred}{}{$1$}
        \draw [->, interconnect,mygreen] (1) to [loop=90,looseness=15] node [jigsaw] {$2$} (1);
        \draw [->, interconnect,mygreen] (2) to [loop=90,looseness=15] node [jigsaw] {$2$} (2);
        \vout{2}{0}{myblue}{}{$3$}
    \end{scope}
  \end{tikzpicture}
  \end{align}  
  This rule makes a partial copy of a node.
  It copies the node including all its loops.
  However, the incoming edges (from the context) and outgoing edges (to the context)
  are not duplicated, but redistributed between the two copies.
  All incoming edges are assigned to the left copy,
  all outgoing edges are assigned to the right copy.
\end{example}

\begin{example}[Non-deterministically Splitting a Node]\label{ex:splitting:node}
  In the preceding example we have seen how to copy
  including all its incoming and outgoing arrows.
  We now want to split a node into two nodes and non-deterministically
  distribute the edges between the two nodes.
  This can be achieved by the following quasi patch graph rewrite rule:
  
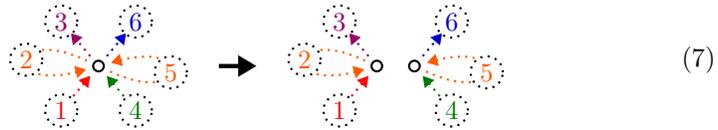
\begin{align}
  \begin{tikzpicture}[default,node distance=5mm,n/.style={graphNode},baseline=0ex]
    \begin{scope}[xshift=0mm]
        \node (1) [n] {};
        \vin{1}{-130}{myred}{}{$1$}
        \draw [->, interconnect,myorange] (1) to [thinloop=180,looseness=25] node [jigsaw] {$2$} (1);
        \vout{1}{130}{mypurple}{}{$3$}
        \vin{1}{-50}{mygreen}{}{$4$}
        \draw [->, interconnect,myorange] (1) to [thinloop=0,looseness=25] node [jigsaw] {$5$} (1);
        \vout{1}{50}{myblue}{}{$6$}
    \end{scope}
    \tikzstep{16}
    \begin{scope}[xshift=37mm]
        \node (1) [n] {};
        \node (2) [n,right of=1] {};
        \vin{1}{-130}{myred}{}{$1$}
        \draw [->, interconnect,myorange] (1) to [thinloop=180,looseness=25] node [jigsaw] {$2$} (1);
        \vout{1}{130}{mypurple}{}{$3$}
        \vin{2}{-50}{mygreen}{}{$4$}
        \draw [->, interconnect,myorange] (2) to [thinloop=0,looseness=25] node [jigsaw] {$5$} (2);
        \vout{2}{50}{myblue}{}{$6$}
    \end{scope}
  \end{tikzpicture}
  \label{rule:split}
  \end{align}
  The patch type of the left-hand side is not a simple graph.
  Here 
  \begin{itemize}
    \item
      \jigsawlabel{1}{myred} and \jigsawlabel{4}{mygreen}
      form a partitioning of the incoming edges (from the context),
    \item
      \jigsawlabel{3}{mypurple} and \jigsawlabel{6}{myblue}
      form a partitioning of the outgoing edges (into the context).
    \item
      \jigsawlabel{2}{myorange} and \jigsawlabel{5}{myorange}
      form a partitioning of the loops on the node.
  \end{itemize}
  There is no fairness condition imposed here. 
  For instance, the partitions 
  \jigsawlabel{1}{myred}, \jigsawlabel{2}{myorange}, \jigsawlabel{3}{mypurple}
  can be empty.
  Then all edges are assigned to the right node. A fair distribution of the edges
  could be achieved by extending the patch types
  with a richer constraint language.
\end{example}

\section{Extended Shorthand Notation}\label{sec:extended:shorthand:notation}

Extending Notation~\ref{notation:shorthand}, we suggest the following notation to indicate that a certain implicit edge should not be present:
      
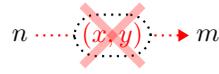
\begin{center}
      \begin{tikzpicture}[default,node distance=25mm,n/.style={graphNode},baseline=(1.base)]
        \begin{scope}[xshift=0mm]
            \node (1) [] {$n$}; 
            \node (2) [right of=1] {$m$};
            \draw [->,interconnect] (1) to node (l) [jigsaw,smalljigsaw] {$(x,y)$} (2);
            \noimplict{l}
        \end{scope}
      \end{tikzpicture}
      \end{center}  
  For instance, rule~\eqref{rule:copy} for copying a node can be written as
  
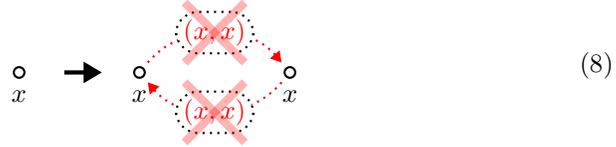
\begin{align}
  \begin{tikzpicture}[default,node distance=20mm,n/.style={graphNode},baseline=0ex]
    \begin{scope}[xshift=0mm]
        \node (1) [n] {};
          \node [anchor=north,at=(1.south),inner sep=0] {$x$};
    \end{scope}
    \tikzstep{6}
    \begin{scope}[xshift=16mm]
        \node (1) [n] {}; 
          \node [anchor=north,at=(1.south),inner sep=0] {$x$};
        \node (2) [n,right of=1] {}; 
          \node [anchor=north,at=(2.south),inner sep=0] {$x$};
        \draw [->,interconnect] (1) to[bend left=50] node (l) [jigsaw,smalljigsaw] {$(x,x)$} (2);
        \noimplict{l}
        \draw [->,interconnect] (2) to[bend left=50] node (l) [jigsaw,smalljigsaw] {$(x,x)$} (1);
        \noimplict{l}
    \end{scope}
  \end{tikzpicture}
  \label{rule:copy:abbreviated}
  \end{align}  
  which is shorthand for
  
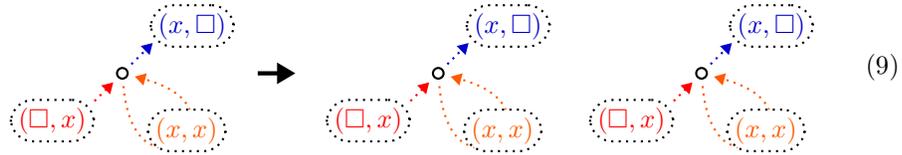
\begin{align}
  \begin{tikzpicture}[default,node distance=35mm,n/.style={graphNode},baseline=0ex]
    \begin{scope}[xshift=0mm]
        \node (1) [n] {};
        \vin{1}{-135}{myred}{smalljigsaw,pos=1.8,xshift=-3mm}{$(\cxtnode,x)$}
        \vout{1}{45}{myblue}{smalljigsaw,pos=1.8,xshift=3mm}{$(x,\cxtnode)$}
        \draw [->, interconnect,myorange] (1) to [loop=-45,looseness=22,pos=0.6] node [jigsaw,smalljigsaw] {$(x,x)$} (1);
    \end{scope}
    \tikzstep{18}
    \begin{scope}[xshift=42mm]
        \node (1) [n] {};
        \node (2) [n,right of=1] {};
        \vin{1}{-135}{myred}{smalljigsaw,pos=1.8,xshift=-3mm}{$(\cxtnode,x)$}
        \vout{1}{45}{myblue}{smalljigsaw,pos=1.8,xshift=3mm}{$(x,\cxtnode)$}
        \draw [->, interconnect,myorange] (1) to [loop=-45,looseness=22,pos=0.6] node [jigsaw,smalljigsaw] {$(x,x)$} (1);
        \vin{2}{-135}{myred}{smalljigsaw,pos=1.8,xshift=-3mm}{$(\cxtnode,x)$}
        \vout{2}{45}{myblue}{smalljigsaw,pos=1.8,xshift=3mm}{$(x,\cxtnode)$}
        \draw [->, interconnect,myorange] (2) to [loop=-45,looseness=22,pos=0.6] node [jigsaw,smalljigsaw] {$(x,x)$} (2);
    \end{scope}
  \end{tikzpicture}
  \end{align}
  Rule~\eqref{rule:split} for non-deterministically splitting a node can now be written as
  
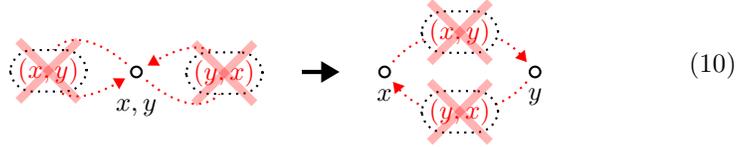
\begin{align}
  \begin{tikzpicture}[default,node distance=20mm,n/.style={graphNode},baseline=0ex]
    \begin{scope}[xshift=0mm]
        \node (1) [n] {};
          \node [anchor=north,at=(1.south),inner sep=0] {$x,y$};
        \draw [->,interconnect] (1) to[loop=180,looseness=22] node (l) [jigsaw,smalljigsaw,pos=0.55] {$(x,y)$} (1);
        \noimplict{l}
        \draw [->,interconnect] (1) to[loop=0,looseness=22] node (l) [jigsaw,smalljigsaw,pos=0.55] {$(y,x)$} (1);
        \noimplict{l}
    \end{scope}
    \tikzstep{22}
    \begin{scope}[xshift=33mm]
        \node (1) [n] {}; 
          \node [anchor=north,at=(1.south),inner sep=0] {$x$};
        \node (2) [n,right of=1] {}; 
          \node [anchor=north,at=(2.south),inner sep=0] {$y$};
        \draw [->,interconnect] (1) to[bend left=50] node (l) [jigsaw,smalljigsaw] {$(x,y)$} (2);
        \noimplict{l}
        \draw [->,interconnect] (2) to[bend left=50] node (l) [jigsaw,smalljigsaw] {$(y,x)$} (1);
        \noimplict{l}
    \end{scope}
  \end{tikzpicture}
  \label{rule:split:abbreviated}
  \end{align}  
  which is an abbreviation for
  
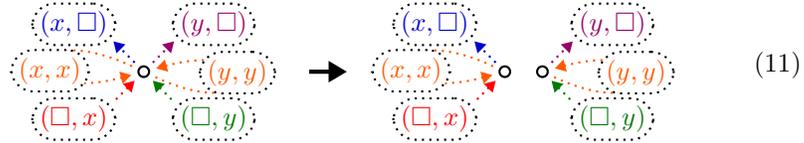
\begin{align}
  \begin{tikzpicture}[default,node distance=5mm,n/.style={graphNode},baseline=0ex]
    \begin{scope}[xshift=0mm]
        \node (1) [n] {};
        \vin{1}{-135}{myred}{smalljigsaw,pos=1.8,xshift=-3mm}{$(\cxtnode,x)$}
        \vout{1}{135}{myblue}{smalljigsaw,pos=1.8,xshift=-3mm}{$(x,\cxtnode)$}
        \vin{1}{-45}{mygreen}{smalljigsaw,pos=1.8,xshift=3mm}{$(\cxtnode,y)$}
        \vout{1}{45}{mypurple}{smalljigsaw,pos=1.8,xshift=3mm}{$(y,\cxtnode)$}
        \draw [->, interconnect,myorange] (1) to [thinloop=180,looseness=35,pos=0.6] node [jigsaw,smalljigsaw] {$(x,x)$} (1);
        \draw [->, interconnect,myorange] (1) to [thinloop=0,looseness=35,pos=0.6] node [jigsaw,smalljigsaw] {$(y,y)$} (1);
    \end{scope}
    \tikzstep{22}
    \begin{scope}[xshift=48mm]
        \node (1) [n] {};
        \node (2) [n,right of=1] {};
        \vin{1}{-135}{myred}{smalljigsaw,pos=1.8,xshift=-3mm}{$(\cxtnode,x)$}
        \vout{1}{135}{myblue}{smalljigsaw,pos=1.8,xshift=-3mm}{$(x,\cxtnode)$}
        \vin{2}{-45}{mygreen}{smalljigsaw,pos=1.8,xshift=3mm}{$(\cxtnode,y)$}
        \vout{2}{45}{mypurple}{smalljigsaw,pos=1.8,xshift=3mm}{$(y,\cxtnode)$}
        \draw [->, interconnect,myorange] (1) to [thinloop=180,looseness=35,pos=0.6] node [jigsaw,smalljigsaw] {$(x,x)$} (1);
        \draw [->, interconnect,myorange] (2) to [thinloop=0,looseness=35,pos=0.6] node [jigsaw,smalljigsaw] {$(y,y)$} (2);
    \end{scope}
  \end{tikzpicture}
  \end{align} 
  Compare the left-hand side of this rule with the right-hand side of the example given in rule~\eqref{rule:merge:abbreviated}.

\section{Termination Detection: Dijkstra--Scholten}\label{sec:terminiation:dijkstra:scholten}

The Dijkstra--Scholten algorithm \cite{fokkink2018distributed} detects the termination of a centralized basic algorithm. It is assumed that the basic algorithm is executed on a loop-free undirected network, and that there is a distinguished initiator process. The Dijkstra--Scholten algorithm builds a tree alongside the execution of the basic algorithm, where processes that are active in the basic algorithm are part of the Dijkstra--Scholten tree. Each process maintains a counter, originally 0. A counter represents a conservative estimate of how many of the process's children are active. When the initiator's counter is $0$, it can quit the tree, by which it correctly announces that the algorithm has terminated (i.e., there are no more basic messages in transit, and all processes have quit the tree).

The algorithm can be described as follows:
\begin{itemize}
\item The initiator starts as the root of the Dijkstra--Scholten tree $T$.
\item A process $p \in T$ can send a basic message to a neighbour process. When it does, it increments its own counter.
\item When a $p \in T$ receives a basic message from a $p'$, it sends $p'$ a control message to indicate that it is already in the tree.
\item When a $p \notin T$ receives a basic message from a $p'$, it joins the tree, and stores $p'$ as its parent.
\item A non-initiator with a counter of $0$ can quit the tree. When it does, it informs its parent that it is no longer its child via a control message.
\item When a process receives a control message, it decrements its counter.
\item An initiator with a counter of $0$ can quit the tree, by which it announces that the basic algorithm has terminated.
\end{itemize}

In our graph representation modeling the state of a Dijkstra--Scholten tree, we use the following labels on directed edges (assume  $u \neq v$):
\begin{itemize}
    \item An edge $u \xrightarrow{e} v$ denotes a network connection between $u$ and $v$.
    \item An edge $u \xrightarrow{b} v$ denotes a basic message in transit between $u$ and $v$.
    \item An edge $u \xrightarrow{c} v$ denotes a control message in transit between $u$ and $v$.
    \item An edge $u \xrightarrow{p} v$ denotes that $v$ has stored $u$ as its parent.
    \item A loop $u \xrightarrow{i} u$ denotes that $u$ is the initiator.
    \item A loop $u \xrightarrow{t} u$ denotes that $u$ is in the Dijkstra--Scholten tree.
    \item A loop $u \xrightarrow{s} u$ denotes an increment of the implicit counter at $u$, which should be interpreted to be $0$ if there are no $s$-loops. 
\end{itemize}

The starting undirected network is encoded as the smallest multigraph containing:
\begin{itemize}
    \item $u \xrightarrow{e} v$ iff there is an undirected edge between $u$ and $v$, and
    \item $u \xrightarrow{i} u$ and $u \xrightarrow{t} u$ iff $u$ is the initiator of the basic algorithm.
\end{itemize}

\begin{notation}[Black Node Shorthand]\label{notation:black:node:shorthand}
As an abbreviation, we use black nodes in rules. For the black nodes and the context node we induce the largest patch type edge subgraph possible (i.e., the most permissive subgraph of patch type edges). The patch type edges are simply copied over to the right-hand side, where nodes that are in the same relative position on the left and on the right are identified. 
\end{notation}

The notation is best explained through an example.

\begin{example}
The rule
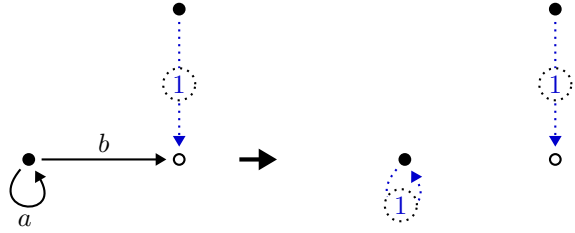
\begin{align*}
  \begin{tikzpicture}[default,node distance=20mm,n/.style={graphNode},baseline=0ex]
  \begin{scope}[xshift=0mm]
    \node (1) [n, fill=black] {};
    \node (2) [n, right of=1] {};
    \node (3) [n, above of=2, fill=black] {};
    \draw [->] (1) to [loop=-90] node [below] {$a$} (1);
    \draw [->] (1) to node [above] {$b$} (2);
    \draw [->, interconnect,myblue] (3) to [] node [jigsaw] {$1$} (2);
  \end{scope}
  \tikzstep{28}
  \begin{scope}[xshift=50mm]
    \node (1) [n, fill=black] {};
    \node (2) [n, right of=1] {};
    \node (3) [n, above of=2, fill=black] {};
    \draw [->, interconnect, myblue] (1) to [loop=-90] node [jigsaw] {$1$} (1);
    \draw [->, interconnect,myblue] (3) to [] node [jigsaw] {$1$} (2);
  \end{scope}
 \end{tikzpicture}
\end{align*}
abbreviates the rule
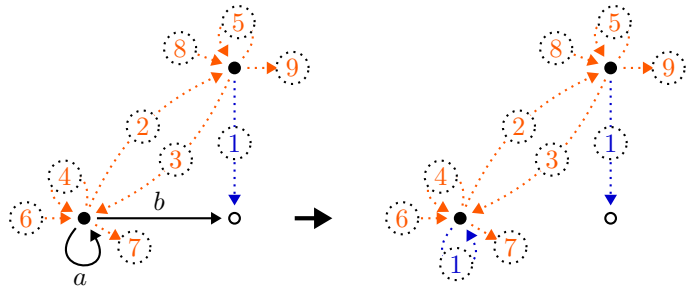
\begin{align*}
  \begin{tikzpicture}[default,node distance=20mm,n/.style={graphNode},baseline=0ex]
  \begin{scope}[xshift=0mm]
    \node (1) [n, fill=black] {};
    \node (2) [n, right of=1] {};
    \node (3) [n, above of=2, fill=black] {};
    \draw [->] (1) to [loop=-90] node [below] {$a$} (1);
    \draw [->] (1) to node [above] {$b$} (2);
    \draw [->, interconnect,myblue] (3) to [] node [jigsaw] {$1$} (2);
    \draw [->, interconnect,myorange] (1) to [bend left=20] node [jigsaw] {$2$} (3);
    \draw [->, interconnect,myorange] (3) to [bend left=20] node [jigsaw] {$3$} (1);
    \draw [->, interconnect, myorange] (1) to [loop=120] node [jigsaw] {$4$} (1);
    \draw [->, interconnect, myorange] (3) to [loop=90] node [jigsaw] {$5$} (3);
    \vin{1}{180}{myorange}{}{$6$}
    \vout{1}{-30}{myorange}{}{$7$}
    \vin{3}{160}{myorange}{}{$8$}
    \vout{3}{0}{myorange}{}{$9$}
  \end{scope}
  \tikzstep{28}
  \begin{scope}[xshift=50mm]
    \node (1) [n, fill=black] {};
    \node (2) [n, right of=1] {};
    \node (3) [n, above of=2, fill=black] {};
    \draw [->, interconnect, myblue] (1) to [loop=-90] node [jigsaw] {$1$} (1);
    \draw [->, interconnect,myblue] (3) to [] node [jigsaw] {$1$} (2);
    \draw [->, interconnect,myorange] (1) to [bend left=20] node [jigsaw] {$2$} (3);
    \draw [->, interconnect,myorange] (3) to [bend left=20] node [jigsaw] {$3$} (1);
    \draw [->, interconnect, myorange] (1) to [loop=120] node [jigsaw] {$4$} (1);
    \draw [->, interconnect, myorange] (3) to [loop=90] node [jigsaw] {$5$} (3);
    \vin{1}{180}{myorange}{}{$6$}
    \vout{1}{-30}{myorange}{}{$7$}
    \vin{3}{160}{myorange}{}{$8$}
    \vout{3}{0}{myorange}{}{$9$}
  \end{scope}
  \end{tikzpicture}
\end{align*}
\end{example}

Observe that Notation~\ref{notation:black:node:shorthand} provides a simpler alternative to Notation~\ref{notation:shorthand} for use cases where node deletion and merging do not take place. We believe it to be generally useful for modeling distributed algorithms, where it is often only the relations between processes that change.

An execution of the Dijkstra--Scholten algorithm can now be modeled using the following rules:
\begin{itemize}
\item Sending a basic message:
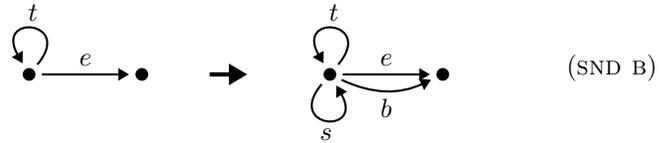
\begin{align}\tag{\textsc{snd b}}
  \begin{tikzpicture}[default,node distance=15mm,n/.style={graphNode},baseline=0ex]
  \begin{scope}[xshift=0mm]
    \node (1) [n, fill=black] {};
    \node (2) [n, right of=1, fill=black] {};
    \draw [->] (1) to [loop=90] node [above] {$t$} (1);
    \draw [->] (1) to node [above] {$e$} (2);
  \end{scope}
  \tikzstep{24}
  \begin{scope}[xshift=40mm]
    \node (1) [n, fill=black] {};
    \node (2) [n, right of=1, fill=black] {};
    \draw [->] (1) to [loop=90] node [above] {$t$} (1);
    \draw [->] (1) to [loop=-90] node [below] {$s$} (1);
    \draw [->] (1) to node [above] {$e$} (2);
    \draw [->] (1) to[bend left=-25] node [below] {$b$} (2);
  \end{scope}
  \end{tikzpicture}\label{dijkstra:scholten:send}
\end{align}

\item Receiving a basic message while in the tree:
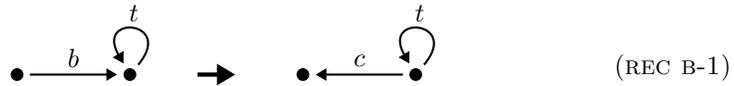
\begin{align}\tag{\textsc{rec b-1}}
  \begin{tikzpicture}[default,node distance=15mm,n/.style={graphNode},baseline=0ex]
  \begin{scope}[xshift=0mm]
    \node (1) [n, fill=black] {};
    \node (2) [n, right of=1, fill=black] {};
    \draw [->] (2) to [loop=90] node [above] {$t$} (2);
    \draw [->] (1) to node [above] {$b$} (2);
  \end{scope}
  \tikzstep{24}
  \begin{scope}[xshift=38mm]
    \node (1) [n, fill=black] {};
    \node (2) [n, right of=1, fill=black] {};
    \draw [->] (2) to [loop=90] node [above] {$t$} (2);
    \draw [->] (2) to node [above] {$c$} (1);
  \end{scope}
  \end{tikzpicture}\label{dijkstra:scholten:rec:in:tree}
\end{align}

\item Receiving a basic message while not in the tree:
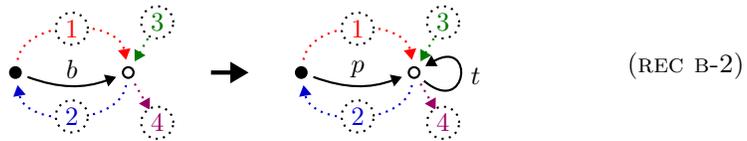
\begin{align}\tag{\textsc{rec b-2}}
  \begin{tikzpicture}[default,node distance=15mm,n/.style={graphNode},baseline=0ex]
  \begin{scope}[xshift=0mm]
    \node (1) [n, fill=black] {};
    \node (2) [n, right of=1] {};
    \draw [->] (1) to [bend right=20] node [above] {$b$} (2);
    \draw [->, interconnect,myred] (1) to [bend left=80] node [jigsaw] {$1$} (2);
    \draw [->, interconnect,myblue] (2) to [bend left=80] node [jigsaw] {$2$} (1);
    \vin{2}{60}{mygreen}{}{$3$}
    \vout{2}{-60}{mypurple}{}{$4$}
  \end{scope}
  \tikzstep{26}
  \begin{scope}[xshift=38mm]
    \node (1) [n, fill=black] {};
    \node (2) [n, right of=1] {};
    \draw [->] (1) to [bend right=20] node [above] {$p$} (2);
    \draw [->, interconnect,myred] (1) to [bend left=80] node [jigsaw] {$1$} (2);
    \draw [->, interconnect,myblue] (2) to [bend left=80] node [jigsaw] {$2$} (1);
    \vin{2}{60}{mygreen}{}{$3$}
    \vout{2}{-60}{mypurple}{}{$4$}
    \draw [->] (2) to [loop=0] node [right] {$t$} (2);
  \end{scope}
  \end{tikzpicture}\label{dijkstra:scholten:not:in:tree}
\end{align}

\item Receiving a control message:
\begin{align}\tag{\textsc{rec c}}
  \begin{tikzpicture}[default,node distance=15mm,n/.style={graphNode},baseline=0ex]
  \begin{scope}[xshift=0mm]
    \node (1) [n, fill=black] {};
    \node (2) [n, fill=black, right of=1] {};
    \draw [->] (2) to node [above] {$c$} (1);
    \draw [->] (1) to [loop=180] node [left] {$s$} (1);
  \end{scope}
  \tikzstep{22}
  \begin{scope}[xshift=35mm]
    \node (1) [n, fill=black] {};
    \node (2) [n, fill=black, right of=1] {};
  \end{scope}
  \end{tikzpicture}\label{dijkstra:scholten:rec:control}
\end{align}

\item A non-initiator quits:
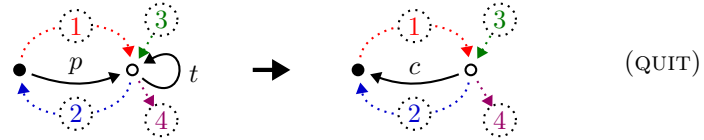
\begin{align}\tag{\textsc{quit}}
  \begin{tikzpicture}[default,node distance=15mm,n/.style={graphNode},baseline=0ex]
  \begin{scope}[xshift=0mm]
    \node (1) [n, fill=black] {};
    \node (2) [n, right of=1] {};
    \draw [->] (1) to[bend right=20] node [above] {$p$} (2);
    \draw [->, interconnect,myred] (1) to [bend left=80] node [jigsaw] {$1$} (2);
    \draw [->, interconnect,myblue] (2) to [bend left=80] node [jigsaw] {$2$} (1);
    \vin{2}{60}{mygreen}{}{$3$}
    \vout{2}{-60}{mypurple}{}{$4$}
    \draw [->] (2) to [loop=0] node [right] {$t$} (2);
  \end{scope}
  \tikzstep{31}
  \begin{scope}[xshift=45mm]
    \node (1) [n, fill=black] {};
    \node (2) [n, right of=1] {};
    \draw [->] (2) to[bend left=20] node [above] {$c$} (1);
    \draw [->, interconnect,myred] (1) to [bend left=80] node [jigsaw] {$1$} (2);
    \draw [->, interconnect,myblue] (2) to [bend left=80] node [jigsaw] {$2$} (1);
    \vin{2}{60}{mygreen}{}{$3$}
    \vout{2}{-60}{mypurple}{}{$4$}
  \end{scope}
  \end{tikzpicture}\label{dijkstra:scholten:noninitiator:quits}
\end{align}

\item The initiator quits/announces:
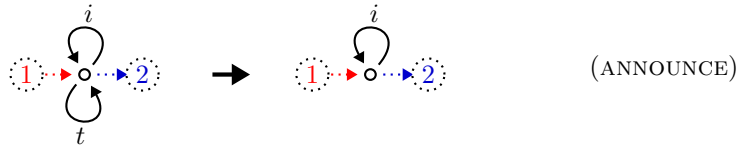
\begin{align}\tag{\textsc{announce}}
  \begin{tikzpicture}[default,node distance=15mm,n/.style={graphNode},baseline=0ex]
  \begin{scope}[xshift=0mm]
    \node (1) [n] {};
    \vin{1}{180}{myred}{}{$1$}
    \vout{1}{0}{myblue}{}{$2$}
    \draw [->] (1) to [loop=90] node [above] {$i$} (1);
    \draw [->] (1) to [loop=-90] node [below] {$t$} (1);
  \end{scope}
  \tikzstep{17}
  \begin{scope}[xshift=38mm]
    \node (1) [n] {};
    \vin{1}{180}{myred}{}{$1$}
    \vout{1}{0}{myblue}{}{$2$}
    \draw [->] (1) to [loop=90] node [above] {$i$} (1);
  \end{scope}
  \end{tikzpicture}\label{dijkstra:scholten:initiator:announces}
\end{align}
\end{itemize}

\section{Modeling Constraints}\label{sec:richer:constraints}

  Instead of using unlabeled type edges, one could label the type edges with expressions from some suitable constraint language and strengthen the definition of adherence. Constraints that immediately suggest themselves include those that restrict the permitted number of adherent edges or their labels.

  For many application scenarios, however, the number of edges is either unbounded, or there is only a small number of permitted choices. The latter case can in principle be handled with unlabeled type edges as follows. Consider the following left-hand side of some rule:
  \begin{center}
  \begin{tikzpicture}[default,node distance=15mm,n/.style={graphNode}]
    \begin{scope}[xshift=0mm]
        \node (1) {$1$};
        \vin{1}{-180}{myred}{draw=none,minimum size=0}{}
    \end{scope}
  \end{tikzpicture}
  \end{center}
  which allows $1$ to have any number of incoming edges from nodes other than itself (and no outgoing edges are allowed).
  Assume instead that we want to express that node $1$
  has either one or two incoming edges labeled with $b$ from nodes other than itself.
  Then we can replace the scheme with three schemes that together capture precisely this constraint:
  
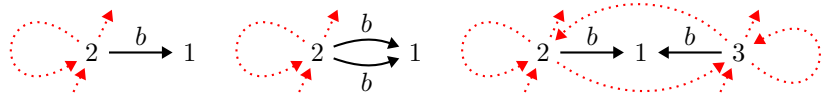
\begin{center}
  \begin{tikzpicture}[default,node distance=13mm,n/.style={graphNode}]
    \begin{scope}[xshift=0mm]
        \node (1) {$1$};
        \node (2) [left of=1] {$2$};
        \draw [->] (2) to node [above] {$b$} (1); 
        \vin{2}{-115}{red}{draw=none,minimum size=0}{} \vout{2}{65}{red}{draw=none,minimum size=0}{}
        
        \draw [->, interconnect,red] (2) to [loop=180,looseness=14] node [] {} (2);
    \end{scope}
    \begin{scope}[xshift=30mm]
        \node (1) {$1$};
        \node (2) [left of=1] {$2$};
        \draw [->] (2) to[bend left=20] node [above] {$b$} (1); 
        \draw [->] (2) to[bend right=20] node [below] {$b$} (1); 
        \vin{2}{-115}{red}{draw=none,minimum size=0}{} \vout{2}{65}{red}{draw=none,minimum size=0}{}
        
        \draw [->, interconnect,red] (2) to [loop=180,looseness=14] node [] {} (2);
    \end{scope}    \begin{scope}[xshift=60mm]
        \node (1) {$1$};
        \node (2) [left of=1] {$2$};
        \node (3) [right of=1] {$3$};
        \draw [->] (2) to node [above] {$b$} (1);
        \draw [->] (3) to node [above] {$b$} (1);
        \vin{2}{-115}{red}{draw=none,minimum size=0}{} \vout{2}{65}{red}{draw=none,minimum size=0}{}
        \vin{3}{-115}{red}{draw=none,minimum size=0}{} \vout{3}{65}{red}{draw=none,minimum size=0}{}
        
        \draw [->, interconnect,red] (2) to [loop=180,looseness=14] node [] {} (2);
        \draw [->, interconnect,red] (3) to [loop=0,looseness=14] node [] {} (3);
        
        \draw [->, interconnect,red] (2) to[bend right=35]  node  {} (3);
        \draw [->, interconnect,red] (3) to[bend right=45]  node  {} (2);
    \end{scope}
  \end{tikzpicture}
  \end{center}
  Clearly, this transformation quickly leads to a prohibitive large number of rules.
  We have chosen to keep things simple as this suffices for many application scenarios.

\end{document}